\documentclass[11pt]{article}

\usepackage{amssymb,graphics,amsmath,amsthm,amsopn,amstext,amsfonts,color,fullpage,fancybox,hyperref,bm}
\usepackage{subfig,placeins} 
\usepackage{graphicx}
\usepackage{verbatim}
\usepackage{algorithm}
\usepackage{algpseudocode}
\usepackage{longtable}
\usepackage{multirow}
\usepackage{enumitem}
	\setenumerate[0]{label=(\alph*)}
	
\newcommand{\gap}{\vspace{0.1in}}
\newcommand{\epc}{\hspace{1pc}}
\newcommand{\thalf}{{\textstyle{\frac{1}{2}}}}
\newcommand{\wh}{\widehat}
\newcommand{\wt}{\widetilde}

\newcommand{\ball}{\mbox{I\!B}}
\newcommand{\re}{\mathbb{R}}
\newcommand{\dist}{\mbox{dist}}
\newcommand{\sgn}{\mbox{sgn }}
\newcommand{\argmin}{\mbox{argmin}}
\newcommand{\argmax}{\mbox{argmax}}
\DeclareMathOperator{\aff}{aff}
\DeclareMathOperator{\rint}{rint}
\DeclareMathOperator{\cl}{cl}

\theoremstyle{definition}
\newtheorem{theorem}{Theorem}
\newtheorem{proposition}[theorem]{Proposition}
\newtheorem{lemma}[theorem]{Lemma}
\newtheorem{corollary}[theorem]{Corollary}

\newtheorem{example}[theorem]{Example}

\title{Exact Penalization of Generalized Nash Equilibrium Problems\footnote{
Both authors are affiliated with the Daniel J.\ Epstein Department of Industrial and Systems Engineering,
University of Southern California, Los Angeles, California 90089-0193, U.S.A. {\tt Emails: qba@usc.edu, jongship@usc.edu.}
This work was based on research supported by the National Science Foundation under grant IIS-1632971 and by the Air Force
Office of Scientific Research under Grant Number FA9550-18-1-0382.}}

\author{Qin Ba \and Jong-Shi Pang}

\date{Original November 18, 2018}	

\begin{document}

\maketitle

\begin{abstract}		
\noindent This paper presents an exact penalization theory of the generalized Nash equilibrium problem (GNEP) that has its origin from the
renowned Arrow-Debreu general economic equilibrium model.  While the latter model is the foundation of much of mathematical economics,
the GNEP provides a mathematical model of multi-agent non-cooperative competition that has found many contemporary applications in diverse engineering
domains. 	The
most salient feature of the GNEP that distinguishes it from a standard non-cooperative (Nash) game is that each player's optimization
problem contains constraints that couple all players' decision variables.  Extending results for stand-alone optimization problems,
the penalization theory aims to convert the GNEP into
a game of the standard kind without the coupled constraints, which is known to be more readily amenable to solution methods and analysis.
Starting with an illustrative example to motivate the development, the paper focuses on two kinds of coupled constraints,
shared (i.e., common) and finitely representable.   Constraint residual functions and the associated error bound theory
play an important role throughout the development.
\end{abstract}

\section{Introduction}

The Generalized Nash Equilibrium Problem (GNEP) extends the classical Nash Equilibrium Problem (NEP) by allowing individual players' constraints,
in addition to the payoff functions, to depend on rival players' strategies.  This extension adds considerably descriptive and explanatory power for the GNEP
to model non-cooperative competition among multiple selfish decision-making agents.
Since its introduction in the 1950's \cite{debreu-1952-social-equil,arrow1954existence} (where the terms {\sl social equilibrium} and
{\sl abstract economy} were used),
the GNEP has risen beyond its mathematical-economic origin and become a common paradigm for various engineering disciplines.
This model fits well in a non-cooperative game setting where all players share some common resources or limitations.  Thus some of the players' constraints
are \emph{coupled}; these are in addition to the players' \emph{private} constraints that do not contain the rival players' decision variables.
Sources of the GNEP include infrastructure systems
such as communication or radio systems \cite{pang2008distributed,pang2010design}, power grids \cite{jing1999spatial,hobbs2007Nash}, transportation
networks \cite{bassanini2002allocation}, modern traffic systems with e-hailing services \cite{BanDessoukyPang2018},
supply and demand constraints for transportation systems \cite{stein2018noncooperative},
and pollution quota for environmental application \cite{Krawczyk05,breton2006game}.
We refer the readers to \cite{facchinei09GNEP,facchinei201012,fischer2014generalized}
for more detailed surveys on this active research topic and its many applications.  The four chapters by Facchinei published
in \cite[pages~151--188]{CominettiFacchineiLasserre12} contain a historical account of the GNEP, provide many references, and
summarize the advances up to the year 2012.

\gap

Due to the presence of the coupling constraints, solving a GNEP is a significantly more challenging task than solving a NEP.
Depending on the types of the coupling constraints, a GNEP can be divided into two categories.  The first category is
the {\sl shared-constrained} case, where the coupling constraints are common to all players.  [In the literature, this case has
also been called jointly constrained, or common constrained; we adopt the ``shared-constrained'' terminology in this paper.]
Many engineering problems cited in last paragraph
are of this type.  The second type is more general in which the coupling constraints can be specific to each player and different among the players.
Due to the relative simplicity and direct connections to many engineering problems, shared-constrained GNEPs receive most of the attention in the literature.
Particularly, much effort has been spent on obtaining a special kind of equilibrium solutions, called {\sl variational equilibrium}, under various settings
and assumptions.  Such an equilibrium has its origin from \cite{rosen1965existence}; the term ``variational equilibrium'' was coined
in \cite{facchinei09GNEP} and further explored in many papers that include \cite{FacchineiFischerPicciali07,facchinei201012,facchinei09GNEP,KulkarniShanbhag12}
where additional  references can be found.  In turn, the adjective
``variational'' was employed to express that this kind of equilibria of a GNEP can be described as solutions of a variational
inequality \cite{harker1991generalized,facchinei2007finite} obtained by concatenating the first-order conditions of the players' optimization problems.
Alternatively, a Nikaido-Isoda-type
bi-function \cite{NikaidoIsoda55} has been used extensively \cite{KrawczykUryasev00,Krawczyk07,von2009optimization} as an attempt to formulate
a GNEP as a single optimization problem.
In general, the existence of a variational equilibrium often relies on the convexity and certain boundedness assumptions and suitable constraint
qualifications \cite{facchinei201012}. 
While computationally simpler and providing a special kind
of solutions, a variational equilibrium does not always have a clear advantage over a non-variational equilibrium solution.
It is hence difficult to argue that special focus should be placed on finding such a particular equilibrium.
Nevertheless, computing a non-variational equilibrium solution is more complex and less understood.  Yet there has been some literature
on this subject.  For example, an extension of the class of variational equilibria is studied in \cite{fukushima2011restricted, nabetani2011parametrized},
and algorithms aiming at obtaining the entire set of solutions of special class of GNEPs are considered 
in \cite{facchinei2011computation, dreves2017computing, aussel2017sufficient}.
The paper \cite{SchiroPangShanbhag13} discusses the numerical solution of an affine GNEP by Lemke's method with an attempt to compute generalizations of
variational equilibria.

\gap

With the goal of exploiting the advances in the NEPs, this paper revisits and sharpens an exact penalization theory for GNEPs
whereby the coupled constraints, shared or otherwise, are moved from the constraints to the objective functions via constraint residual functions
scaled by a penalty parameter.  Dating back to the beginning of nonlinear programming \cite{FiaccoMcCormick68}, exact penalization for
optimization has been a subject of intensive study in the early days of optimization (see \cite{HanMangasarian79}) and much research
has been undertaken by the Italian School led by Di Pillo, Facchinei and Grippo
\cite{demyanov1998exact, di1989exact, di1992regularity, di1995exact, di1989b-exact}.  Exemplified by these works and references therein, this theory deals largely
with problems whose constraints are defined by differentiable inequalities and relies on the multipliers of these constraints.
The papers \cite{ByrdLopezCalvaNocedal12,ByrdNocedalWaltz08}
discuss how penalty functions can be used to design numerical methods for solving nonlinear programs that have been implemented in the
very successful {\sc knitro} software \cite{ByrdNocedalWaltz06}.  For the GNEP, exact penalty results are derived in
\cite{facchinei2006exact, facchinei-2011-Lampariello} which also discuss algorithmic design.  The paper \cite{facchinei-2010-penal-method}
discusses extensive computational rules for updating the penalty parameter and reports numerical results with the resulting algorithms.
Contrary to exact penalization, inexact, or asymptotic penalization also penalizes the constraints but the penalty parameter
is required to tend to infinity to recover a solution of the given problem.  A fairly simple iterative inexact-penalty based
method for solving the GNEP was discussed in \cite{pang-2009-quasi-variat}; this method cast the game as a {\sl quasi variational inequality}
obtained by concatenating the first-order optimality conditions of the players' nonlinear programs.  The recent papers \cite{KanzowSteck16,KanzowSteck18}
expanded the latter paper and discussed augmented Lagrangian and exact penalty methods for the GNEP and the more general quasi-variational inequality,
respectively.

\gap

In contrast to the multiplier-based exact penalization, the principle of exact penalization, first
established in \cite[Proposition~2.4.3]{Clarke83} and subsequently generalized in \cite[Theorem~6.8.1]{facchinei2007finite}
deals with the exact penalization of an abstract constraint set.  The latter result makes it clear that this
theory is closely connected to the theory of error bounds \cite{bauschke-1999-stron-conic,lewis1998error} and
\cite[Chapter~6]{facchinei2007finite}.  To date, this approach has not been applied to the GNEP; part of the goal
of this paper is to establish how exact penalty results for the GNEP can be derived based on an error-bound theory.  Unlike
a stand-alone optimization problem, a GNEP has multiple inter-connected optimization problems each with an objective
function of its own; thus while exact penalty results for the players' optimization problems are available, simply putting
them together does not readily yield an exact penalty result for a GNEP.  The inter-connection between the players' optimization problems
will need to be accounted for when the coupling constraints are penalized. Furthermore, while most literature on the exact penalization of 
the GNEP considers only coupling inequality constraints, 
see e.g., \cite{facchinei-2010-penal-method, facchinei-2011-Lampariello, fukushima2011restricted, KanzowSteck16}; 
another goal of this paper is to bring inequality coupling constraints to the study of exact penalization. 
By means of an illustrative example, we show that
care is needed in choosing the constraint residual functions in order for the conversion to a penalized NEP to be valid.  This example
also paves the way for the theory to be developed subsequently.  It is worth noting that the \( \ell_{q} \) residual functions for
\( q \in (1, \infty) \) present certain advantages in guaranteeing exact penalization of GNEP.  This class of residual functions is used
in \cite{facchinei-2011-Lampariello} and \cite[Chapter~3 Section~4.1]{CominettiFacchineiLasserre12}.  This paper also takes advantage of them in finitely representable cases;
Sections~\ref{sec:share-constraint} and \ref{sec:finitely representable} therein elucidate the advantageous functional
properties of \( \ell_{q} \) residual functions in exact penalization.

\section{Preliminary Discussion and Review} \label{sec:GNEP}

In this section, we present the setting of this paper and fix the notations to be used throughout.
We begin with the definition of the GNEP, then review some relevant concepts and the principle of exact penalization for an optimization
problem, and finally end with an illustration of this principle applied to a simple GNEP and a preliminary result.

\subsection{The GNEP}

Formally, a GNEP consists of \( N \) selfish players each of whom decides a strategy with the goal to minimize (or maximize) his/her cost (or payoff).
Let \( x^{\nu} \in \re^{n_{\nu}} \) denote the decision of player \( \nu \in [N] \triangleq \{ 1, \cdots, N \} \).
Anticipating the rival players' strategies \( x^{-\nu} \triangleq ( x^{\nu^{\, \prime}} )_{\nu^{\, \prime} \neq \nu} \),
player \( \nu \in [N] \) solves the following minimization problem in his/her own variable $x^{\nu}$:
\begin{equation} \label{opt:gnep}
\begin{array}{ll}
\displaystyle{
\operatornamewithlimits{\mbox{minimize}}_{x^{\nu}\in C^{\, \nu}}
} & \theta_{\nu}(x^{\nu}, x^{-\nu}) \\ [5pt]
\mathrm{subject to} & x^{\nu} \, \in \, X^{\nu}(x^{-\nu}),
\end{array}
\end{equation}
where \( \theta_{\nu}(x^{\nu}, x^{-\nu}) \) denotes the cost function of player \( \nu \) as a function of the pair $( x^{\nu},x^{-\nu} )$;
the constraints are of two types: those described by the
private strategy set $C^{\, \nu} \subseteq\re^{n_{\nu}}$ and the others described by the coupled constraint
set $X^{\nu}(x^{-\nu}) \subseteq\re^{n_{\nu}}$ that depends on the rival players' strategy tuple $x^{-\nu}$.
[Writing \( (x^{\nu}, x^{-\nu}) \) in the argument of \( \theta_{\nu}(\bullet) \)  is to emphasize player $\nu$'s strategy
\( x^{\nu} \) and does not mean that the tuple \( x \triangleq (x^{\nu})_{\nu =1}^{N} \in \re^{n} \), where
$n \triangleq \displaystyle{
\sum_{\nu \in [N]}
} \, n_{\nu}$, is reordered such that \( x^{\nu} \)
is the first block.]
This setup clearly indicates that both the objective function and constraint set of each player's optimization problem
are coupled by all players' strategies. 
Let \( \theta(x) \triangleq (\theta_{\nu}(x^{\nu},x^{-\nu}))_{\nu=1}^{N} \) be the concatenation of the players' objective functions;
$C \triangleq \displaystyle{
\prod_{\nu \in [N]}
} \, C^{\, \nu}$ be the concatenation of the players' private constraints,
and  \( X :\re^n \to\re^n \) with $X(x) \triangleq \displaystyle{
\prod_{\nu\in [N]}
} \, X^{\nu}(x^{-\nu})$ be the multifunction concatenating the players' coupled constraints.
We use the triplet \( (C,X;\theta) \) to denote the GNEP with players' optimization problems defined by \eqref{opt:gnep}.
A tuple $x^* \triangleq ( x^{*,\nu} )_{\nu \in [N]}$ is a {\sl Nash Equilibrium} (NE) of this game if
$x^{*,\nu} \in \displaystyle{
\operatornamewithlimits{\mbox{argmin}}_{x^{\nu} \in C^{\, \nu} \, \cap \, X^{\nu}(x^{*,-\nu})}
} \, \theta(x^{\nu},x^{*,-\nu})$ for all $\nu \in [N]$.

\gap

By considering $X^{\nu}$ as a multifunction from $\mathbb{R}^{n_{-\nu}}$ into $\mathbb{R}^{n_{\nu}}$, where
$n_{-\nu} \triangleq n - n_{\nu}$, the constraint $x^{\nu} \in X^{\nu}(x^{-\nu})$ is equivalent to the membership that
$(x^{\nu},x^{-\nu})$ belongs to the graph of $X^{\nu}$, which is denoted by $\mbox{graph}(X^{\nu}) \subseteq\re^n$.
An important special case of the GNEP \( (C,X;\theta) \) is when these graphs are all equal to a common set $D \subseteq \re^{n}$.  This is the
case of common coupled constraints, which we call shared constraints.  We denote this special case by the
triplet \( (C,D;\theta) \).  If each function $\theta_{\nu}(\bullet,x^{-\nu})$ is differentiable and convex and the set
$C  \cap D$ is convex, then a special kind of NE of the GNEP \( (C,D;\theta) \) is that derived from the variational inequality,
denoted by VI $(C \cap D; \Theta)$,
of finding a vector $\bar{x} \in C \cap D$ such that $( y - \bar{x} )^{\mathrm{T}} \Theta(\bar{x}) \geq 0$ for all $y \in C \cap D$, where
$\Theta(x) \triangleq \left( \nabla_{x^{\nu}} \theta_{\nu}(x) \right)_{\nu \in [N]}$.  Such a special NE is a called a {\sl variational equilibrium}.  When the graphs $\mbox{graph}(X^{\nu})$ are not the same, NEs of the GNEP \( (C,X;\theta) \) are related to solutions of a
quasi-variational inequality; see e.g.\ \cite{pang-2009-quasi-variat} for details.

\gap

Another important special case of the GNEP \( (C,X;\theta) \) is when each set $X^{\nu}(x^{-\nu})$ is finitely representable, i.e.,
when
\begin{equation} \label{eq:finitely representable}
X^{\nu}(x^{-\nu}) \, \triangleq \, \left\{ \, x^{\nu} \, \in \,\re^{n_{\nu}} \, \mid \, g^{\, \nu}(x^{\nu},x^{-\nu}) \, \leq \, 0
\mbox{ and } h^{\nu}(x^{\nu},x^{-\nu}) \, = \, 0 \, \right\}
\end{equation}
for some functions $( g^{\, \nu}, h^{\, \nu} ) :\re^n \to\re^{m_{\nu} + p_{\nu}}$ with some positive integers $m_{\nu}$ and $p_{\nu}$.
Typically, such constraint functions
$g^{\, \nu}$ and $h^{\nu}$ are such that $g^{\nu}(\bullet,x^{-\nu})$ is differentiable and each of its component functions $g_i^{\, \nu}(\bullet,x^{-\nu})$ is
convex and $h^{\nu}(\bullet, x^{\nu})$ is affine; in this case, the set $X^{\nu}(x^{-\nu})$ is both closed and convex;
see Section~\ref{sec:finitely representable}.

\gap

{\bf A road map.}  For the GNEP \( (C,X;\theta) \) and its special cases, exact penalization aims to convert these games with coupled constraints
to standard games
whereby the constraint set $X^{\nu}(x^{-\nu})$ for each player is moved to the objective function via a residual function (see definition in the next subsection) scaled
by a penalty parameter $\rho$, thereby obtaining a NEP \( (C;\theta_{\rho}) \) for a concatenated penalized objective $\theta_{\rho}$,
whose solution is a solution of the given GNEP.
Notice that we retain the private constraints described by the set
$C^{\, \nu}$ in each player's optimization problem.  Thus this is a type of {\sl partial} exact penalization \cite{facchinei-2011-Lampariello}.
The development in the rest of this paper is as follows.
We first present an illustrative example which reveals important insights into exact penalization of GNEP, then consider the penalization of the GNEP \( (C,D;\theta) \) with common coupled constraints and general closed convex sets $C$ and $D$ with no particular structure.  The latter part is concluded with the case where $D$ is finitely representable by convex inequalities and linear equalities.  
Finally, we discuss the
GNEP \( (C,X;\theta) \) by penalizing the coupled constraints defined as in \eqref{eq:finitely representable}.  
We should note that the literature on exact penalization of the GNEP has dealt largely with the latter problem under
constraint qualifications, such as the sequentially bounded CQ \cite{facchinei2006exact,pang-2009-quasi-variat} and the Extended Mangasarian-Fromovitz constraint qualification (EMFCQ) \cite{facchinei-2010-penal-method, facchinei-2011-Lampariello, fukushima2011restricted, KanzowSteck16}. In this paper, we have provided results for exact penalization of  GNEP under a strong descent assumption and a Lipschitz error bound assumption for residual functions, see Proposition~\ref{pr:NEP feasible yields GNEP}, Theorem~\ref{th:exact penal shared} and Corollary~\ref{co:exact penalization under sufficient}, and a Slater-type condition for the coupled constraints, see Theorem~\ref{th:exact penal shared finitely representable} and ~\ref{th:finitely representable}.

\subsection{Relevant concepts} \label{sec:concepts}
We introduce in this subsection some relevant concepts that are used later.
A function \( \Phi:\re^{n} \to \re^{m} \) is said to be \emph{Lipschitz continuous} on a set \( S \subseteq \re^{n} \) if,
for some positive scalar \( \mbox{Lip}_{\Phi} \), the following holds:
\begin{equation*}
| \, \Phi(x) - \Phi(y) \, | \, \leq \, \mbox{Lip}_{\Phi} \, \| \, x-y \, \|_{2}, \quad \forall x, y \in S.
\end{equation*}
This is also referred to as Lipschitz continuity with constant \( \mbox{Lip}_{\Phi} \). In addition, \( \Phi \) is said to be Lipschitz continuous 
 {\sl near} a point \( x \)  if, for some \( \varepsilon > 0 \), \( \Phi \) is Lipschitz continuous in the \( \varepsilon \)-neighborhood of \( x \).
We hence say that $\Phi$ is {\sl locally Lipschitz continuous} on $S$ if it is Lipschitz continuous near every point in $S$.
 The {\sl directional derivative} of $\Phi$ at a point \( x \in \re^{n} \) in a direction \( d \in \re^{n} \)
is defined to be
\begin{equation}
\label{eq:def-dd}
\Phi^{\, \prime}(x; d) \, \triangleq \, \displaystyle{
\lim_{ \tau \downarrow 0}
} \, \displaystyle{
\frac{\Phi(x+\tau d) - \Phi(x)}{\tau}
}
\end{equation}
if the limit exists.  If the above limit exists for all \( d\in \re^{n} \), \( \Phi \) is said to be \emph{directionally differentiable} at \( x \).
It is clear that if \( \Phi \) is Lipschitz continuous with constant $\mbox{Lip}_{\Phi}$
near \( \bar{x} \), then 
\( \| \, \Phi^{\, \prime}(\bar{x}; d) \, \| \) (if the derivative exists) is bounded by \( \mbox{Lip}_{\Phi} \, \| \, d \, \| \) for all $d$.


\gap

For two given closed subsets $S$ and $T$ of $\mathbb{R}^n$,
a function $r_S :\re^n \to\re_+$ is a $T$-{\sl residual function} of $S$ if
for all $x \in T$, it holds that $r_S(x) = 0$ if and only if $x \in S$, or equivalently if and only if $x \in S \cap T$.
Note that this definition allows for $S$ and $T$ to have an empty intersection; in this case, the residual function
$r_S$ is positive on $T$.  Such a residual function \( r_{S}(x) \) is said to yield a
$T$-{\sl Lipschitz error bound} of \( S \) if there exist a constant \( \gamma>0 \) such that
\begin{equation*}
\dist(x;S \cap T) \, \leq \, \gamma \, r_S(x), \quad \forall \, x \, \in \, T,
\end{equation*}
where $\dist(x;Z) \triangleq \displaystyle{
\operatornamewithlimits{\mbox{minimum}}_{y \in Z}
} \, \| \, y - x \, \|_2$ is the distance function to a closed set $Z$ under the Euclidean norm.  These definitions extend the usual
concepts of a residual function and associated error bound of a set $S$ which pertain to the case when $T =\re^n$.   In essence,
we are interested in such quantities only for elements in the subset $T$; thus the restriction to the latter set.
%
%
If $S$ is finitely representable defined by the $m$-dimensional vector function $g$ and the $p$-dimensional vector function $h$
as in (\ref{eq:finitely representable}),
we define the $\ell_q$ {\sl residual function} of \( S \) for a given \( q \in [1, \infty] \) as
\begin{equation}
\label{def:l-k-residual}
r_q(x) \, \triangleq \, \left\| \, \left( \begin{array}{c}
g^{+}(x) \\ [3pt]
h(x)
\end{array} \right) \, \right\|_q,
\end{equation}
where $g^+(x) \triangleq \max( g(x),0 )$ and $\| \bullet \|_q$ is the $\ell_q$-norm of vectors.
If each \( g_{i}, i\in [m] \) is convex and \( h \) is affine, then
\( r_q \) is convex for all \( 1\le q \le \infty \).
Moreover, if these defining functions are differentiable, then for
\( q \in (1, \infty) \), \( r_q(x) \) is differentiable at \( x \not \in S \).  Its gradient is given as follows:
\begin{equation}
\label{eq:lk-gradient}
\nabla r_q(x) \, = \, \displaystyle{
\frac{1}{r_q^{q-1}(x)}
} \, \left[ \, \displaystyle{
\sum_{i\in [m]}
} \left( \, g^{+}_{i}(x) \, \right)^{q-1} \, \nabla g_{i}(x) + \displaystyle{
\sum_{j\in [p]}
} \, | \, h_{j}(x) \, |^{q-1} \, ( \, \sgn h_{j}(x) \, ) \, \nabla h_{j}(x) \, \right],
\end{equation}
where \( \sgn(\bullet) \) is the sign function of a scalar defined to be 1 for positive numbers, -1 for negative numbers, and 0 for 0. In the same case, the function \( r_q(x) \) is directional differentiable at all \( x \in S \) and its directional derivative is given by:
\begin{equation} \label{eq:lk-dd}
r_q^{\, \prime}(x;d) \, = \, \left( \sum_{i:\, g_{i}(x) =0 } \, [ \, \max( \, \nabla g_{i}(x)^{\mathrm{T}}d, 0 ) \, ]^q \, + \,
\displaystyle{ \sum_{j:\, h_{j}(x) = 0}
} \, | \, \nabla h_{j}(x)^{\mathrm{T}} d \, |^q  \right)^{1/q},
\end{equation}
As for \( q =1 \), \( r_1(x) \) is not differentiable but directional differentiable and its directional derivative is given by:
\begin{equation} \label{eq:l1-dd}
\begin{aligned}
r_1^{\, \prime}(x; d) \, = & \sum_{i:\, g_{i}(x) >0 } \nabla g_{i}(x)^{\mathrm{T}} d +
\sum_{i: \, g_{i}(x) = 0 } \max( \, \nabla g_{i}(x)^{\mathrm{T}} d, 0 \, ) + \\
& \sum_{j: \, h_{j}(x) =0 } | \, \nabla h_{j}(x)^{\mathrm{T}} d \, | + \sum_{j:\, h_{j}(x) \neq 0 } ( \, \sgn h_{j}(x) \, ) \, \nabla h_{j}(x)^{\mathrm{T}} d.
\end{aligned}
\end{equation}

\subsection{Partial exact penalization}

Exact penalization of GNEPs has its origin from its counterpart for solving constrained optimization problems.
As such, we begin by considering the following general constrained optimization problem:
\begin{equation} \label{opt:nlp}
\displaystyle{
\operatornamewithlimits{\mbox{minimize}}_{x \, \in \, W \, \cap \, S}
} \ f(x),
\end{equation}
where \( W \) and \( S \) represent two sets of constraints with \( S \) considered more complex than \( W \).
By penalizing the violation of membership in $S$, a penalty formulation transforms \eqref{opt:nlp} into the following penalized optimization problem:
\begin{equation}
\label{opt:nlp-penalized}
\displaystyle{
\operatornamewithlimits{\mbox{minimize}}_{x \, \in \, W }
} \ f(x) +  \rho \, r_{S}(x),
\end{equation}
where \( \rho > 0 \) is a \emph{penalty parameter} to be specified and \( r_{S}(x) \) is a $W$-residual function of the set \( S \).
The penalization is said to be \emph{exact} for \eqref{opt:nlp} if there exists a positive scalar \( \bar{\rho} \) such that, for all
\( \rho \geq \bar{\rho} \), every global minimizer of the penalized problem is also a global minimizer of \eqref{opt:nlp} and vice versa.
Note: Dealing with minimizers, this definition is a classical one.  In this regard, the definition is computationally most meaningful
when both (\ref{opt:nlp}) and (\ref{opt:nlp-penalized}) are convex programs.  Otherwise,
the equivalence between the minimizers is a conceptual property that
offers guidance to solve constrained problems in practical algorithms via relaxations of some constraints.
Subsequently, when we apply this optimization theory to the GNEP, we will make some needed convexity assumptions so that we can freely talk
about minimizers.

\gap

It is important to note that the penalized problem (\ref{opt:nlp-penalized}) is different from the usual one as discussed in
\cite[Section~6.8]{facchinei2007finite} which is:
\begin{equation} \label{eq:usual penalization}
\displaystyle{
\operatornamewithlimits{\mbox{minimize}}_{x \, \in \, W }
} \ f(x) +  \rho \, r_{S \cap W}(x),
\end{equation}
where $r_{S \cap W}(x)$ is a residual function of the intersection $S \cap W$ in the unrestricted sense.  As such, while
the result below is similar to \cite[Theorem 6.8.1]{facchinei2007finite}, they differ in several ways: one, the penalized
problems are different; two, the error bound conditions are different; and three, the result below does not pre-assume the
solvability of either (\ref{opt:nlp}) or (\ref{opt:nlp-penalized}).  This is in contrast to the solvability assumption that
is needed for one inclusion of the two argmin sets in the cited theorem.  For completeness, we offer a detailed proof of the
result below which does not explicitly assume convexity of the functions involved.


\begin{theorem} \rm
\label{thm:nlp-exact-penalization}
Consider the nonlinear program \eqref{opt:nlp} with closed sets \( W \) and \( S \) and a Lipschitz continuous objective function $f$
with constant $\mbox{Lip}_f > 0$ on the set $W$.  Assume that $S \, \cap \, W \, \neq \, \emptyset$.
Let \( r_S(x) \) be a $W$-residual function of the set
$S$  satisfying a $W$-Lipschitz error bound for the set $S$ with constant $\gamma > 0$;
i.e., \( r_S(x) \geq \gamma \dist(x;S \cap W) \) for all \( x \in W\).
Then for all \( \rho > \mbox{Lip}_f/ \gamma \),
\begin{equation*}
\underset{x\in S \cap W}{\argmin} \, f(x) \, = \, \underset{x \in W }{\argmin} \, f(x) + \rho \, r_{S}(x).
\end{equation*}
\end{theorem}

\begin{proof}  In the following two-part proof, we may assume without loss of generality that both argmin sets are nonempty.
We first show that the left-hand argmin is a subset of the right-hand argmin.  Let $\bar{x}$ be an element of
$\displaystyle{
\operatornamewithlimits{\mbox{argmin}}_{x \in S \cap W}
} \, f(x)$.  Let $x \in W$ be arbitrary.  Let $\wh{x}$ be a vector in $S \, \cap \, W$ such that $\| x - \wh{x} \|_2 = \dist(x;S \cap W)$.
We have
\[ \begin{array}{llll}
f(x) + \rho \, r_{S}(x) & \geq & f(x) + \rho \, \gamma \, \| \, x - \wh{x} \, \|_2 & \mbox{by the $W$-Lipschitz error bound} \\ [5pt]
& \geq & f(\wh{x}) - \mbox{Lip}_f \, \| \, x - \wh{x} \, \|_2 + \rho \, \gamma \, \| \, x - \wh{x} \, \|_2 & \mbox{by the Lipschitz continuity of $f$} \\ [5pt]
& \geq & f(\wh{x}) & \mbox{by the choice of $\rho$} \\ [5pt]
& \geq & f(\bar{x}) \, = \, f(\bar{x}) + \rho \, r_{S}(\bar{x}) & \mbox{by the optimality of $\bar{x}$}.
\end{array} \]
This establishes one inclusion.  For the converse, let $x^* \in \displaystyle{
\operatornamewithlimits{\mbox{argmin}}_{x \in W}
} \, f(x) + \rho \, r_{S}(x)$.  It suffices to show that $x^* \in S$.  Let $\wt{x}$ be a vector in $S \, \cap \, W$ such that $\| x^{*} - \wt{x} \|_2 = \dist(x^{*};S \cap W)$.
We have
\[ \begin{array}{llll}
f(\wt{x}) & \geq & f(x^*) + \rho \, r_{S}(x^*) & \mbox{by the optimality of $x^*$} \\ [5pt]
& \geq & f(x^*) + \rho \, \gamma \, \| \, x^* - \wt{x} \, \|_2 & \mbox{by the $W$-Lipschitz error bound} \\ [5pt]
& \geq & f(\wt{x}) - \mbox{Lip}_f \, \| \, x^* - \wt{x} \|_2 +  \rho \, \gamma \, \| \, x^* - \wt{x} \, \|_2 & \mbox{by the Lipschitz continuity of $f$} \\ [5pt]
& \geq & f(\wt{x}) & \mbox{by the choice of $\rho$}.
\end{array} \]
Therefore, $x^* = \wt{x}$, establishing the claim and the equality of the two sets of minimizers.
\end{proof}

We will apply Theorem~\ref{thm:nlp-exact-penalization}
to the GNEP \( (C,X; \theta) \) by penalizing the coupled constraint set $X^{\nu}(x^{-\nu})$ in
each player $\nu$'s optimization problem.  Letting $C^{\, -\nu} \triangleq \displaystyle{
\prod_{\nu^{\, \prime} \neq \nu}
} \, C^{\, \nu^{\, \prime}}$ and \( r_{\nu}(\bullet,x^{-\nu}) \) be a $C^{\, \nu}$-residual function of
the set $X^{\nu}(x^{-\nu})$ for given \( x^{-\nu} \in C^{\, -\nu} \), we obtain
the following penalized subproblem for player \( \nu \) with respect to given \( x^{-\nu} \):
\begin{equation} \label{opt:jointly-cvx-game-penalty}
\displaystyle{
\operatornamewithlimits{\mbox{minimize}}_{x^{\nu} \, \in \, C^{\, \nu}}
} \ \theta_{\nu}(x^{\nu},x^{-\nu}) + \rho \, r_{\nu}(x^{\nu}, x^{-\nu})
\end{equation}
where $\rho > 0$ is a penalty parameter which we take for simplicity to be the same for all players.
Let \( r_{X}(x) \triangleq (r_{\nu}(x))_{\nu=1}^{N} \) and $\theta_{\rho;X} \triangleq \theta + \rho \, r_X$.
Then the optimization problems \eqref{opt:jointly-cvx-game-penalty} concatenated
for all $\nu \in [N]$ define the game \( (C;\theta_{\rho;X}) \) which is clearly a NEP.
We say that exact penalization holds for the GNEP \( (C,X;\theta) \) with residual function \( r_{X} \) if there exists a scalar
\( \bar{\rho} > 0 \) such that, for all \( \rho \ge \bar{\rho} \), every equilibrium solution of the NEP \( (C;\theta_{\rho;X}) \)
is also an equilibrium solution of the GNEP \( (C,X;\theta) \) and vise versa.  The following example drawn from \cite{facchinei201012} illustrates this penalization.

\begin{example} \rm
\label{eg:non-variational-equiv}
Consider the shared-constrained GNEP \( (C,D;\theta) \) with the following specifications: $N = 2$; $n_1 = n_2 = 1$;
$\theta_{1}(x) = (x_{1}-1)^{2}$; $\theta_{2}(x) = (x_{2}-1/2)^{2}$; \( C = \re^{2} \), and
\( D= \{x\in \re^{2} \, | \, x_{1} + x_{2} \le 0 \} \).
The set of equilibria of this GNEP is \( \{(\alpha, 1- \alpha)\, |\, 1/2 \le \alpha \le 1\} \).
The set of variational equilibria of the GNEP is a singleton containing \( (3/4, 1/4) \).  With \( \rho>0 \) as the penalty parameter, the
penalized NEP is defined by the following two univariate optimization problems:
\begin{equation*}
\displaystyle{
\operatornamewithlimits{\mbox{minimize}}_{x_1}
} \ (x_{1}-1)^{2} + \rho \, ( \, x_{1}+x_{2}-1 \, )^{+} \quad
\displaystyle{
\operatornamewithlimits{\mbox{minimize}}_{x_2}
} \ \left( \, x_{2} - \thalf \, \right)^{2} + \rho \, ( \, x_{1}+x_{2}-1 \, )^{+}
\end{equation*}
that employ the common residual function $r_D(x) = ( \, x_{1}+x_{2}-1 \, )^{+}$.
The set of equilibria of the latter NEP is obtained by investigating the three cases shown in Table~\ref{tab:three-cases}.
\begin{table}[htbp]
   \centering
   \begin{tabular}{@{} cccc @{}} 
\hline
          & \( x_{1}^{*} \)  & \( x_{2}^{*} \)  & equilibria set for \( \rho \ge 1 \)  \\
\hline
      \( x_{1}+x_{2} > 1 \)       & \( 1-\rho/2 \)  & \( 3/2 - \rho \)  & \( \emptyset \)  \\
       \( x_{1}+x_{2} < 1 \)       & \( 1 \)   & \( 3/2 \)  & \( \emptyset \)  \\
       \( x_{1}+x_{2} = 1 \)       & \( [1-\rho/2, 1] \)   & \( [1/2-\rho/2, 1/2] \)  &  \( \{(\alpha, 1- \alpha)\, |\, 1/2 \le \alpha \le 1\} \)  \\
\hline
   \end{tabular}
   \caption{Three cases of the penalized problem}
   \label{tab:three-cases}
\end{table}
From the table, we conclude that the given GNEP has an exact penalization as a standard NEP.  \hfill $\Box$

\end{example}

The primary question this paper aims at addressing is under what conditions exact penalization holds for the GNEP \( (C,X; \theta) \) and
its special cases.
While various exact penalization results exist for the optimization problem~\eqref{opt:nlp} in the literature, stand-alone
results on the GNEP are not as many; they are mostly embedded in methods for solving the game under some CQs without necessarily
claiming the exactness of the penalization.
A main difficulty in extending the penalty results for an optimization problem to the GNEP lies in the fact that
while each player's optimization problem~\eqref{opt:gnep} may have an exact
penalty equivalent, the penalty parameter for player $\nu$'s problem in principle depends on the rivals' strategy tuple $x^{-\nu}$;
this makes it difficult for a uniform penalty parameter to exist for all players in an NEP formulation.
A related point is highlighted by the proof of the converse inclusion in
Theorem~\ref{thm:nlp-exact-penalization} which suggests that the feasibility of a solution of the penalized NEP is key
for that solution to be a solution of the original un-penalized game.  This point was made clear in \cite[Theorem~1]{facchinei2006exact}
and the discussion that follows it where an explanation was offered.
We formally state this requirement in the first part of the following preliminary result whose proof we omit.

\begin{proposition} \label{pr:NEP feasible yields GNEP} \rm
Suppose that $C$ and the graph of the constraint multifunction $X$ are closed sets.
Let $r_{\nu}(\bullet,x^{-\nu})$ be a $C^{\, \nu}$-residual function of the set $X^{\nu}(x^{-\nu})$,
for all $x^{-\nu} \in C^{\, -\nu}$.  The following two statements hold:
\begin{enumerate}
\item[\rm (a)] Let $x^*$ be a solution of the penalized NEP \( (C;\theta_{\rho;X}) \) for some $\rho > 0$.  Then
$x^*$ is a solution of the GNEP $(C,X;\theta)$
if and only if $x^*$ belongs to the graph of the multifunction $X$; i.e., $x^{*,\nu} \in X^{\nu}(x^{*,-\nu})$ for all
$\nu \in [ N ]$.
\item[\rm (b)]
Suppose that for every $x^{-\nu} \in C^{\, -\nu}$, $C^{\, \nu} \, \cap \, X^{\nu}(x^{-\nu}) \neq \emptyset$, 
$\theta_{\nu}(\bullet,x^{-\nu})$ is Lipschitz continuous with constant $\mbox{Lip}_{\nu}$ on the set $C^{\, \nu}$,
and $r_{\nu}(\bullet,x^{-\nu})$ provides a $C^{\, \nu}$-Lipschitz error bound with constant $\gamma_{\nu} > 0$
for the set $X^{\nu}(x^{-\nu})$.
Then for all $\rho > \displaystyle{
\max_{\nu \in [N]}
} \, \mbox{Lip}_{\nu}/\gamma_{\nu}$, every equilibrium solution of the penalized NEP \( (C;\theta_{\rho;X}) \)
is an equilibrium solution of the GNEP $(C,X;\theta)$  and vice versa.   \hfill $\Box$
\end{enumerate}
\end{proposition}

One of the main weakness of part (b) of the above proposition is the assumption that $C^{\, \nu} \, \cap \, X^{\nu}(x^{-\nu})$ is nonempty
for all $x^{-\nu}$ in $C^{\, -\nu}$.  The invalidity of this assumption contributes to inequality between the set of
penalized NE and that of the NE of the GNEP; see parts (a) of Example~\ref{eg:3-penalty-funs} below.
Subsequently, we will provide alternative conditions that bypass this assumption;
see Theorem~\ref{th:exact penal shared} and~\ref{th:finitely representable}.

\section{An Illustrative Example} \label{sec:illustrative}

Before presenting the main results, we offer an example that summarizes several important features in the penalization of
the shared-constrained GNEP \( (C,D; \theta) \) by residual functions of the common constraint set $D$.
This example provides the basis for the subsequent results that are motivated by the various penalized NEPs discussed therein.

\begin{example} \label{eg:3-penalty-funs}
Consider the share-constrained GNEP \( (C,D; \theta) \) with the following specifications: $N = 2$; $n_1 = n_2 = 1$; \( \theta_{1}(x) = -x_{1} \);
\( \theta_{2}(x) = -x_{2} \); \( C = [0, 4] \times [0, 4] \); \( D= \{ x\in \re^{2} \, |\, g_{i}(x) \le 0, i=1, 2, 3 \} \); where
\( g_{1}(x) = x_{1} + x_{2} - 2 \), \( g_{2}(x) = 2 x_{1} - x_{2} - 2 \), and \( g_{3} = -x_{1} + 2x_{2} -2 \).
The constraint set \( C \cap D \) is shown as the gray area in Figure~\ref{fig:eg-3-penalty-funs}.
It is easy to see that every point on the line segment between \( (4/3, 2/3) \) and \( (2/3, 4/3) \) (i.e., the solid dark line segment in
the figure) is a NE of this GNEP, for which we illustrate various penalizations as follows.

\begin{figure}[htbp] 
   \centering
   \includegraphics[width=0.4\linewidth]{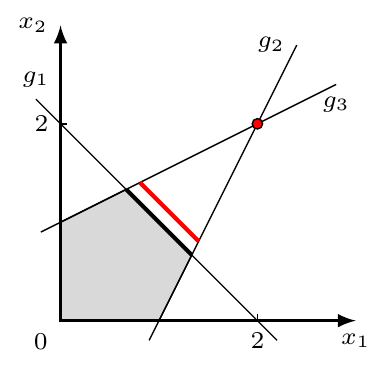}
   \caption{Graph illustration for Example~\ref{eg:3-penalty-funs}}
   \label{fig:eg-3-penalty-funs}
\end{figure}
\end{example}

(a) {\bf \( \ell_{1} \) penalization is not exact.}  Using the \( \ell_{1} \) residual function \( r_1(x) = \| g^{+}(x) \|_{1} \),
the penalized version of the GNEP \( (C,D; \theta) \) is the NEP \( (C; \theta + \rho \, r_1) \).
While the \( \ell_{1} \) residual function is a widely used residual function that is proven to be effective for exact penalization of many
nonlinear programming problems, this example shows that it does not work generally for exact penalization of even shared, linearly constrained GNEPs.
Specifically, we can find an equilibrium solution of the penalized NEP \( (C; \theta + \rho \, r_1) \) but not one of the GNEP \( (C,D; \theta) \).  The two players' optimization problems for the latter NEP 
are as follows:
\begin{gather*}
\displaystyle{
\operatornamewithlimits{\mbox{minimize}}_{0\le x_1\le 4}
} \  -x_{1} + \rho \, \left( \, [  x_{1} + x_{2} - 2 ]^{+} + [ 2 x_{1} - x_{2} - 2 ]^{+} + [ -x_{1} + 2x_{2} -2 ]^{+} \, \right) \\
\displaystyle{
\operatornamewithlimits{\mbox{minimize}}_{0\le x_2 \le 4}
} \   -x_{2} + \rho \, \left( \, [  x_{1} + x_{2} - 2 ]^{+} + [ 2 x_{1} - x_{2} - 2 ]^{+} + [ -x_{1} + 2x_{2} -2 ]^{+} \, \right).
\end{gather*}
By examining the different pieces of the square box \( C \) separated by the three lines in the plane, we can conclude that, for \( \rho > 1 \),
the set of equilibria of the NEP \( (C; \theta + \rho \, r_1) \) is \( \{ (x_{1}, x_{2}) \,|\, x_{1}+x_{2} =2, \, 2/3 \le x_{1} \le 4/3 \} \cup (2, 2) \).
It is clear that \( (2, 2) \) is not an equilibrium solution of the given GNEP \( (C,D; \theta) \).

\gap

This negative result is very informative; namely it suggests that unlike the case of an optimization problem,
one needs to be selective in the penalty function in the case of the GNEP; see Theorem~\ref{th:exact penal shared}.

\gap

(b) {\bf Squared \( \ell_2 \) penalization is not exact}. 
The penalized problem with the squared quadratic residual function
\( r_2^2(x) = \| g^{+}(x) \|_2^2 \) is the NEP \( (C; \theta + \rho \, r^{2}_{2}) \).  It is easy to verify that for any \( \rho >0 \),
every point in \( \{x\in C\,|\, g_{1}(x) = 1/(2\rho), \, g_{2}(x) \le 0, \, g_{3}(x)\le 0\} \) (i.e., the red line segment in Figure~\ref{fig:eg-3-penalty-funs})
is an equilibrium solution of this penalized NEP.  
Even though as \( \rho \) gets to infinity, these points converge to equilibria of the original GNEP, for any finite \( \rho \) they do not belong to \( C \cap D \), and thus are not equilibria of the original GNEP.

\gap

(c) {\bf \( \ell_{2} \) penalization is exact.} The penalized problem with the \( \ell_{2} \) residual function is the NEP \( (C; \theta + \rho \, r_{2}) \).
By \eqref{eq:lk-gradient}, at \( x\not\in D \), the gradient of \( r_{2}(x) \) is:
\begin{equation*}
\nabla r_{2}(x) = \frac{1}{r_{2}(x)} \left(
\begin{array}{c}
g_{1}^{+}(x) + 2 g^{+}_{2}(x)- g^{+}_{3}(x) \\ [5pt]
g_{1}^{+}(x) -  g^{+}_{2}(x) + 2 g^{+}_{3}(x)
\end{array} \right)
\end{equation*}
For any \( x \)  in the interior of \( C \setminus D \), one can write down the equilibrium conditions of this NEP as
\[
\left( \begin{array}{c}
-1 \\ [3pt]
-1
\end{array} \right) + \, \rho \, \nabla r_2(x) \, = \, 0
\]
which can be shown to have no solution for $\rho > 1$. Similarly, by checking the respective optimality conditions of the two players' optimization problems, one can conclude that any point
on the boundary of \( C \setminus D \) 
 is not be an equilibrium solution of the NEP \( (C; \theta + \rho \, r_{2}) \) either.
Furthermore, by  part (a) of Proposition~\ref{pr:NEP feasible yields GNEP},
it is straightforward to see that any point of \( C \cap D \) that is not an equilibrium solution of the GNEP \( (C,D; \theta) \)
cannot be one of the NEP \( (C; \theta + \rho \, r_{2}) \).  The set of remaining feasible points of NEP \( (C; \theta + \rho \, r_{2}) \) is the equilibria set of the GNEP \( (C,D; \theta) \), i.e., \( \{ (x_{1}, x_{2}) \,|\, x_{1}+x_{2} =2, 2/3 \le x_{1} \le 4/3 \} \).  Using the directional derivative formula given in \eqref{eq:lk-dd},
it is possible to verify that every point in this set is an equilibrium solution of the penalized NEP \( (C; \theta + \rho \, r_{2}) \).

\gap

(d) {\bf \( \ell_{1} \) penalization is exact for variational equilibria.} In this example, every equilibrium solution of
the GNEP \( (C,D; \theta) \) is a variational equilibrium and hence a solution to the VI \( (C \cap D; \Theta) \).  By penalizing the constraint set \( D \)
of this VI using the \( \ell_{1} \) residual function (same as part (a)), we obtain the mutivalued VI \( (C; \Theta + \rho \, \partial r_{1}) \),
where  $\partial$ denotes the subdifferential of a convex function as in convex analysis.
It is not difficult to see that the solution set of the penalized VI agrees with the solution set of VI \( (C \cap D; \Theta ) \),
and hence with the equilibrium set of the GNEP \( (C,D; \theta) \), for all \( \rho \ge 1 \).  This conclusion should be contrasted with part (a)
where the penalization of this GNEP is via the NEP there.  In terms of it first-order conditions, the latter NEP is equivalent to the
multivalued VI \( (C; \Theta + \rho \, R) \), where
$R$ is the multifunction $R(x) = \displaystyle{
\prod_{i=1}^3
} \, \partial_{x_i} r_1(x)$ which is a superset of $\partial r_1(x)$.  Thus it is not unexpected that the NEP may have more solutions than the
given GNEP, as confirmed by part (a). 

\gap

(e) {\bf \( \ell_{1} \) penalization is exact when $C$ is restricted.}  As a result of the failed exact penalization
of the GNEP \( (C,D; \theta) \) in terms of the NEP $(C; \theta + \rho \, r_1)$ demonstrated in part (a), one could raise the question
of whether the given GNEP may be equivalent to a NEP $( \wt{C}, \theta + \rho \, r_1)$ for a certain Cartesian subset $\wt{C}$ of $C$.
It turns out that this question has an affirmative answer for this example. Indeed, a natural way to define this restricted set is
$\wt{C} = \wt{C}_1 \times \wt{C}_2$ with
\[
\wt{C}_1 \, \triangleq \, \left\{ \, x_1 \, \in \, C_1 \, \mid \, \exists \, x_2 \, \in \, C_2 \mbox{ such that } ( x_1,x_2 ) \, \in \, D \, \right\}
\]
and similarly for $\wt{C}_2$.  For this example, we have
\( \wt{C}_{1} = \wt{C}_{2} = [0, 4/3] \).  We leave it for the reader to verify that
\[ \begin{array}{rl}
&\mbox{set of equilibrium solutions of the NEP \( (\wt{C}; \theta + \rho \, r_{1}) \)} \\ [5pt]
\epc = \, &\mbox{set of equilibrium solutions of the GNEP \( ( \wt{C}, D; \theta) \)} \\ [5pt]
\epc = \, &\mbox{set of equilibrium solutions of the GNEP \( (C, D; \theta ) \)}.
\end{array} \]
The equality between the first and second set illustrates part (a) of Proposition~\ref{pr:NEP feasible yields GNEP}. It is possible to extend the above observation to more general cases. The formal study on this is to be pursued in the future.
\hfill $\Box$

\gap

An important difference between the $\ell_1$ and $\ell_2$ penalization is that the residual function in the latter case
is differentiable at points outside the shared constraint set $D$
whereas that in the former is not.   This observation motivates the following consideration about the choice of the residual function
$r_D$.  Namely, including the distance function to the set $D$ when the latter is closed and convex, the residual function $r_D$ needs to be at the
minimum directionally differentiable.  However, this alone is not enough as illustrated by part (a) in the above example.  Thus a suitable
condition on the directional derivatives outside the set $D$ is needed; such a requirement is made explicit in Theorem~\ref{th:exact penal shared} in the next section.

\section{Main Results for the GNEP $(C,D;\theta)$} \label{sec:share-constraint}

Due to the common shared constraint set $D$, we employ a common residual function that applies to all players' individual optimization problems;
see Examples~\ref{eg:non-variational-equiv} and \ref{eg:3-penalty-funs} for such a function.
Specifically, throughout this section, we let $r_D :\re^n \to\re_+$ be a $C$-residual function of the shared constrained set $D$;
i.e., for all $x \in C$, $r_D(x) = 0$ if and only if $x \in D$.
It then follows that for each $x^{-\nu} \in C^{-\, \nu}$, $r_D(\bullet,x^{-\nu})$ is
a $C^{\, \nu}$-residual function of the set $D^{\, \nu}(x^{-\nu}) \triangleq \{ y^{\nu} \mid ( y^{\nu},x^{-\nu} ) \in D \}$.
Conversely, we can construct such an (aggregate) residual function $r_D$
from the latter sets $D^{\, \nu}(x^{-\nu})$.  Indeed, for each $\nu \in [ N ]$, let $r_{\nu} :\re^n \to\re_+$ be such that
for each $x^{-\nu} \in C^{-\, \nu}$, $r_{\nu}(\bullet,x^{-\nu})$ is a $C^{\, \nu}$-residual function of the set $D^{\, \nu}(x^{-\nu})$. It is
then easy to verify that with $r_D(x) \triangleq \displaystyle{
\sum_{\nu \in [ N ]} } \, r_{\nu}(x^{\nu},x^{-\nu})$, it holds that for all $x \in C$, $r_D(x) = 0$ if and only if $x \in D$.
According to part (a) of Proposition~\ref{pr:NEP feasible yields GNEP}, a key requirement for a solution of the penalized
NEP $(C;\theta + \rho \, r_D)$, where $\theta + \rho \, r_D$ is a shorthand for the vector function $( \theta_{\nu} + \rho \, r_D )_{\nu \in [ N ]}$,
to be a solution of the GNEP $(C,D;\theta)$ is the membership of the candidate solution on hand in the shared constraint set $D$.
Moreover, as mentioned above, the directional differentiability of the residual function $r_D$ needs to be strengthened.
Taking into account these two considerations and adapting a key quantity introduced in \cite{demyanov1998exact},
we present the following Strong Descent Assumption on residual functions. It would be shown to play a vital role in guaranteeing exact penalization of a GNEP.
For a set  \( S \subseteq \re^{n} \) and a vector \( x \in S \), let \( \mathcal{T}(x; S) \) denote the {\sl tangent cone} at \( x \)
to  \( S \), i.e., the collection of vectors \( d \in \re^{n} \) such that \( d = \displaystyle{
\lim_{k\to \infty}
} \, \displaystyle{
\frac{x^{k} -x}{\tau_{k}}
} \) for some vector sequence \( \{x^{k}\} \subset S \) converging to \( x \) and some positive scalar sequence \( \{\tau_{k}\} \)
with \( \tau_{k} \downarrow 0 \).

\gap

{\bf Strong Descent Assumption}. Let \( S \) and \( W \) be two closed sets and \( r_{S} \)  be a \( W \)-residual function of \( S \).
There exists a constant \( \alpha>0 \) such that for every \( x \in W\setminus S \), \( r_{S}^{\, \prime}(x; d) \le -\alpha \|d\|_{2} \)
for some nonzero \( d \in \mathcal{T}(x; W) \).  \hfill $\Box$

\gap

The above assumption appears in \cite[Theorem~3.2]{demyanov1998exact} in the context of a single optimization problem.
While it will be discussed in some more detail momentarily, the assumption essentially requires that a strong descent direction
of the residual function exists at every infeasible point \( x\in W\setminus S \).  Considering the optimization
problem \eqref{opt:nlp} with Lipschitz continuous function \( f(x) \), the strong descent assumption would disqualify
every point in \( W\setminus S \) as a minimizer of \eqref{opt:nlp-penalized} for sufficiently large \( \rho \).
Specialized to each player \( \nu \)'s optimization problem and to the associated residual function \( r_{D}(\bullet, x^{-\nu}) \),
this argument leads to the next result on the exact penalization of the GNEP  \( (C, D; \theta) \),
where convexity of the sets \( C \), \( D \) and the functions \( \theta_{\nu}(\bullet, x^{-\nu}) \) is not needed.

\begin{theorem}  \label{th:exact penal shared}
Suppose that \(C\), \( D \) are closed subsets of \( \re^{n} \). Assume that each \( \theta_{\nu}(\bullet, x^{-\nu}) \) is directionally differentiable and locally Lipschitz continuous with constant \( \text{Lip}_{\theta} >0 \) on \( C^{\, \nu} \) for all \( x^{-\nu} \in C^{\, -\nu} \).
Let \( r_{D} \) be a directionally differentiable \( C \)-residual function of \( D \).  If there exist positive constants $\alpha$ and $\alpha^{\, \prime}$
such that for every \( x \in C\setminus D\), either one of the following holds:
\begin{enumerate}[label= (\alph*)]
\item for some \( \nu \in [N] \) and nonzero \( d^{\, \nu} \in \mathcal{T}(x^{\nu}; C^{\, \nu}) \subseteq \re^{n_{\nu}} \),
\begin{equation} \label{eq:negative-dd}
r_{D}(\bullet, x^{-\nu})^{\, \prime}(x^{\nu}; d^{\, \nu}) \, \leq \, - \alpha^{\, \prime} \, \| \, d^{\, \nu} \, \|_{2}
\end{equation}
\item for some nonzero \( d \in \mathcal{T}(x; C) \),
\begin{equation} \label{eq:bounded-dd}
\displaystyle{
\sum_{\nu \in [N]}
}  \, r_D(\bullet, x^{-\nu})^{\,\prime}(x^{\nu}; d^{\, \nu}) \, \leq \, r_{D}^{\,\prime}(x; d) \, \leq \, - \alpha \, \| \, d \, \|_{2},
\end{equation}
\end{enumerate}
then there exists a finite number \( \bar{\rho} \) such that for every \( \rho > \bar{\rho} \), every equilibrium solution of the
NEP \( (C; \theta + \rho r_{D}) \) is an equilibrium solution of the GNEP \( (C, D; \theta) \).
\end{theorem}
\begin{proof}  Since $\| \, d \, \|_2 \, \geq \, \displaystyle{
\frac{1}{\sqrt{N}}
} \, \displaystyle{
\sum_{\nu \in [N]}
} \, \| \, d^{\, \nu} \, \|_2$, it follows that if (b) holds for some $\alpha$, then (a) holds with $\alpha^{\, \prime} = \alpha/\sqrt{N}$.
Hence it suffices to prove the theorem under condition (a).
Let \( \bar{x} \) be an equilibrium solution of the
NEP \( (C; \theta + \rho \, r_D) \). By part (a) of Proposition~\ref{pr:NEP feasible yields GNEP}, it suffices to show that $\bar{x} \in D$.
Assume for contradiction that \( \bar{x} \in C \setminus D \). By assumption, there exist
\( \nu \in [N] \) and \( d^{\, \nu}\in \mathcal{T}(\bar{x}^{\nu}; C^{\, \nu}) \) such that
\[
r_{D}(\bullet, \bar{x}^{-\nu})^{\, \prime}(\bar{x}^{\nu}; d^{\, \nu}) \, \leq \, - \alpha^{\, \prime} \, \| \, d^{\, \nu} \, \|_{2}.
\]
By the optimality of $\bar{x}^{\nu}$ to player $\nu$'s optimization problem in the NEP \( (C; \theta + \rho \, r_D) \),
we have 
\begin{equation*}
\theta_{\nu}(\bullet, \bar{x}^{-\nu})^{\prime}(\bar{x}^{\nu};d^{\, \nu}) +
\rho \, r_{D}(\bullet,\bar{x}^{-\nu})^{\prime}(\bar{x}^{\nu};d^{\, \nu}) \, \ge \, 0
\end{equation*}
By the Lipschitz continuity of $\theta_{\nu}(\bullet,x^{-\nu})$ with constant $\mbox{Lip}_{\theta}$, we obtain
\[
| \,  \theta_{\nu}(\bullet, \bar{x}^{-\nu})^{\prime}(\bar{x}^{\nu};d^{\, \nu}) \, |
\, \leq \, \mbox{Lip}_{\theta} \, \| \, d^{\, \nu} \, \|_2.
\]
Hence, It follows that
\[
\mbox{Lip}_{\theta} \, \| \, d^{\, \nu} \, \|_2 - \rho \, \alpha^{\, \prime} \, \| \, d^{\nu} \, \|_{2}  \, \geq \, 0.
\]
By choosing of $\bar{\rho} =\triangleq \mbox{Lip}_{\theta}/ \alpha^{\, \prime}$, the above inequality yields a contradiction.
%
%
\end{proof}

Theorem~\ref{th:exact penal shared} is a game-theoretic extension of Theorem~\ref{thm:nlp-exact-penalization}
which pertains to a single optimization problem. It is useful to note that the assumption of a Lipschitz
error bound on the residual function required in the former theorem is replaced by the strong descent assumptions in the present
game result. Furthermore, in contrary to part (b) of Proposition~\ref{pr:NEP feasible yields GNEP}, Theorem~\ref{th:exact penal shared}
does not require $C^{\, \nu} \cap D^{\, \nu}(x^{-\nu})$ to be nonempty for all \( x^{-\nu} \in C^{-\nu} \).  As shown in
Example~\ref{eg:3-penalty-funs}, the later nonemptiness condition may be too restrictive even for a simple linearly constrained GNEP. We also note that part (a) of Theorem~\ref{th:exact penal shared} can be extended in a straightforward way to the general case of GNEP \( (C, X; \theta) \). The assumptions required there would involve the strong descent assumption of \( r_{\nu}(\bullet, x^{-\nu}) \). While this result is not explicitly stated in the paper, another sufficient condition would be provided for the GNEP \( (C, X; \theta) \) in the case that \( X_{\nu}(\bullet) \) is finitely representable, see Theorem~\ref{th:finitely representable}.

\gap

In the remaining of this section, due to the relatively simpler condition, we will focus on refining part (b) of Theorem~\ref{th:exact penal shared}
by providing some sufficient conditions for \eqref{eq:bounded-dd} to hold.
The inequality on the left side of \eqref{eq:bounded-dd} is the sum property on the total directional derivative
$r_D^{\, \prime}(x;y - x)$ with regard to the partial directional derivatives $r_D(\bullet,x^{-\nu})^{\prime}(x^{\nu};y^{\, \nu} - x^{\nu})$ for $\nu \in [N]$.
The study of this relationship for a general function is originated from \cite{Auslender76} (see also \cite{Tseng02}).
It is clear that when \( r_{D} \) is convex, this inequality is equivalent to equality as the directional derivative
of a convex function is subadditive in the direction.  Sufficient conditions for the equality include G\^{a}teaux differentiability, and a certain
strong property of the directional derivatives that can be found in \cite{Robinson91}.  In particular, if the residual function $r_D$ is differentiable on
$C \setminus D$ (to be precise, differentiable on a set $\Omega \, \setminus \, D$, where
$\Omega$ is a open set containing the set $C$), then the left side inequality in \eqref{eq:bounded-dd} holds readily.
This differentiability assumption will be in place in the rest of the section.

\gap

The inequality on the right side of \eqref{eq:bounded-dd} is related to the inequality (15) imposed in Theorem~3.2 in \cite{demyanov1998exact}
that is for the total penalization of a single optimization problem.  The cited theorem dealt with the optimization problem $\displaystyle{
\operatornamewithlimits{\mbox{minimize}}_{x \in D}} \ \theta(x)$ where $D$ is a subset of $\mathbb{R}^n$ and assumed
in essence to satisfy the property that there exist positive scalars $\alpha$ and $\varepsilon$
such that for all $x \in D_{\varepsilon} \setminus D$, some $y \in\re^n \setminus \{ x \}$ exists satisfying
$r_D^{\, \prime}(x;y - x) \, \leq \, -\alpha \, \| \, y - x \, \|$, where $D_{\varepsilon}$ is an  \( \varepsilon \)-enlargement of the set $D$;
there was no second set $C$ involved.  For partial penalization, the incorporation of the set \( C \) in the strong descent assumption appears
to be new.  While this restriction reduces the set of potential descent directions from \( \re^{n} \) to \( \mathcal{T}(x; C) \),
it also excludes certain infeasible points under consideration, thus facilitating the exact penalization.
We will provide a more detailed discussion on this point in subsequent subsections.


\subsection{Linear metric regularity} \label{subsec:lmr}

The literature on Lipschitz error bounds for closed convex sets is vast.  Summarizing previous papers that include \cite{pang1997error,lewis1998error,Klatte1999} and
many others, Chapter~6 in the monograph \cite{facchinei2007finite} provides a good reference for this theory and its applications.  Closely related to
Lipschitz error bounds is the theory of linear metric regularity; see \cite{bauschke-1999-stron-conic,ng2004-regul-their}.  For our purpose, we can
rephrase the strong descent assumption in terms of the latter theory.
Specifically, two closed convex subsets $C$ and $D$ of $\mathbb{R}^n$ with a nonempty intersection are said to be {\sl linearly metrically regular}
if there exists a constant $\gamma^{\, \prime} > 0$ such that
\[
\dist(x;C \cap D) \, \leq \, \gamma^{\, \prime} \, \max\left( \, \dist(x;C), \, \dist(x;D) \, \right), \epc \forall \, x \, \in \,\re^n.
\]
We remark that linear metric regularity is equivalent to the existence of a constant $\gamma_1 > 0$ such that
\begin{equation} \label{eq:equivalent lmr}
\dist(x;C \cap D) \, \leq \, \gamma_1 \, \dist(x; D), \epc \forall \, x \, \in \, C.
\end{equation}
The above equivalence is known in the literature, e.g., see \cite[Theorem 3.1]{ng2004-regul-their}. Yet we provide a proof here for completeness. Clearly, the former implies the latter with $\gamma_1 = \gamma$.  If the latter holds, let $x \in\re^n$ be arbitrary and let
$y$ be the Euclidean projection of $x$ onto $C$.  We then have
\[ \begin{array}{lll}
\dist(x;C \cap D) & \leq & \dist(x;C) + \dist(y;C \cap D) \\ [5pt]
& \leq & \dist(x;C) + \gamma_1 \, \dist(x; D) \, \leq \, ( \, \gamma_1 + 1 \, ) \,
\max\left( \, \dist(x;C), \, \dist(x;D) \, \right),
\end{array} \]
establishing the equivalence of linear metric regularity and the inequality (\ref{eq:equivalent lmr}).  It follows from this equivalence
that if $C$ and $D$ are linearly metrically regular, then a $C$-residual function $r_D$ of the set $D$ satisfies a $C$-Lipschitz error bound
if there exists a constant $\gamma > 0$ such that $\dist(x;D) \leq \gamma \, r_D(x)$ for all $x \in C$, the latter condition being
a usual error bound for $D$ with reference to vectors in $C$.

\gap

The result below shows that linear metric regularity (or equivalently (\ref{eq:equivalent lmr})) provides a sufficient condition for the strong descent assumption to hold.

\begin{proposition} \label{pr:linear metric implies error} \rm
Suppose the two closed convex sets $C$ and $D$ with nonempty intersection are linearly metrically regular with constant $\gamma^{\, \prime}$.  Let $r_D$
be a convex residual function of $D$ satisfying a $C$-Lipschitz error bound for $D$ with constant $\gamma > 0$.  If $r_D$
is differentiable on the set $C \setminus D$, then \eqref{eq:bounded-dd} holds.
\end{proposition}

\begin{proof} It remains to show the right side of \eqref{eq:bounded-dd} holds with $\alpha = \displaystyle{
\frac{1}{\gamma \, \gamma^{\, \prime}}
}$.  Let $x \in C \setminus D$ and $y$ be the Euclidean projection of $x$ onto the intersection $C \cap D$.  By the convexity of
the residual function $r_D$, we have
\[ \begin{array}{lll}
r_D^{\, \prime}(x;y - x) & \leq & r_D(y) - r_D(x) \, = \, -r_D(x) \, \leq \, - \gamma \, \mbox{dist}(x;D) \\ [0.1in]
& \leq & - \displaystyle{
\frac{\gamma}{\gamma^{\, \prime}}
} \, \mbox{dist}(x;C \, \cap \, D) \, = \, -\displaystyle{
\frac{\gamma}{\gamma^{\, \prime}}
} \, \| \, x - y \, \|_2
\end{array}
\]
as desired.
\end{proof}

In addition to the necessary and sufficient inequality (\ref{eq:equivalent lmr}), there are well-known sufficient conditions
for two closed convex sets to be linearly metrically regular.  In particular, we provide in Proposition~\ref{pr:sufficient for lmr}
below a Slater condition (part (a)) and a polyhedrality condition (part (b)).
The proof of the first statement of the proposition follows from \cite[Corollary 5, Remark 9 and 10]{bauschke-1999-stron-conic} and
and that of the second statement from the well-known Hoffman error bound for polyhedra \cite[Lemma 3.2.3]{facchinei2007finite}.
In the result, we use the shorthand ``rint'' to denote the relative interior of a convex set.

\begin{proposition} \label{pr:sufficient for lmr} \rm
Two closed convex sets $C$ and $D$ in $\mathbb{R}^n$ are linearly metrically regular under either one of the following two conditions:
\begin{description}
\item[\rm (a)] $C \cap D$ is bounded and $\mbox{rint}(C) \, \cap \, \mbox{rint}(D) \, \neq \, \emptyset$;
\item[\rm (b)] $C$ and $D$ are both polyhedra. \hfill $\Box$
\end{description}
\end{proposition}

Combining part (b) of Theorem~\ref{th:exact penal shared} and Propositions~\ref{pr:linear metric implies error} and \ref{pr:sufficient for lmr},
we have the following exact penalization result for the GNEP $(C,D;\theta)$ in terms of the penalized NEP $(C,\theta + \rho \, r_D)$, for which no proof is needed.

\begin{corollary} \label{co:exact penalization under sufficient} \rm
Let $D$ be a closed convex set in $\mathbb{R}^n$ and $C \triangleq \displaystyle{
\prod_{\nu \in [ N ]}
} \, C^{\, \nu}$ where each $C^{\, \nu}$ is a closed convex subset of $\mathbb{R}^{n_{\nu}}$.
Assume that each $\theta_{\nu}(\bullet,x^{-\nu})$ is convex and locally Lipschitz continuous on the set
$C^{\, \nu}$ with constant $\mbox{Lip}_{\theta} > 0$ for all $x^{-\nu} \in C^{\, -\nu}$.
Let $r_D$ be a convex $C$-residual function of the set $D$ that is differentiable
on $C \setminus D$.  Then under either one of the following two conditions, there exists a scalar $\bar{\rho} > 0$ such that for all
$\rho > \bar{\rho}$,  every equilibrium solution of the NEP \( (C; \theta + \rho \, r_{D}) \) is an equilibrium solution of the GNEP \( (C, D; \theta) \):
\begin{description}
\item[\rm (a)] $r_D$ satisfies a $C$-Lipschitz error bound for the set $D$; $C \cap D$ is bounded;
and $\mbox{rint}(C) \, \cap \, \mbox{rint}(D)$ is nonempty;
\item[\rm (b)] $C$ and $D$ are both polyhedra. \hfill $\Box$
\end{description}
\end{corollary}

Part (b) of Corollary~\ref{co:exact penalization under sufficient} provides a formal generalization of part (c) of the illustrative Example~\ref{eg:3-penalty-funs}. It also demonstrates that, for shared linearly constrained games and with the use of
\( \ell_{q} \) residual functions for \( q\in (1, \infty) \), the Lipschitz error bound condition required in part (a) is no longer needed. In the next subsection, we will show that this is true in general for convex finitely representable set \( D \).

\subsection{Shared finitely representable sets} \label{subsec:shared finitely representable}

In this subsection and the rest of the paper, we assume that the set $C$ is compact and convex.
We discuss the GNEP $(C,D;\theta)$ where the shared-constrained set $D$ is finitely representable by coupled differentiable
inequalities and linear equations.  Specifically, let
\begin{equation} \label{eq:shared finitely representable X}
D \, \triangleq \, \left\{ \, x \, \in \, \mathbb{R}^n \, \mid \,
g(x) \, \leq \, 0 \mbox{ and } h(x) \, = \, 0 \, \right\}
\end{equation}
with $h(x) = Ax - b$ for some matrix $A \in \mathbb{R}^p$ and vector $b \in \mathbb{R}^p$ and each
$g_j : \mathbb{R}^n \to \mathbb{R}$ for $j \in [m]$ being convex and differentiable.
We separate the linear and nonlinear constraints in $D$ so that we can state the
Slater condition more precisely; see assumption in Lemma~\ref{lm:shared Slater}.
We employ the $\ell_q$ residual function for an arbitrary $q \in (1,\infty)$ for the set $D$:
\[
r_q(x) \, \triangleq \, \left\| \, \left( \, \begin{array}{c}
\max( \, g(x), \, 0 \, ) \\ [3pt]
h(x)
\end{array} \right) \, \right\|_q, \epc x \, \in \, \mathbb{R}^n.
\]
The function $r_q$ is continuously differentiable on the complement of the set $D$.
Toward the formulation of an exact penaliztion result for the shared constrained GNEP $(C,D;\theta)$ with
a finitely representable set $D$, we establish
a consequence of the Slater CQ for the set $D$. This lemma is related to many results in the literature; but surprisingly, we cannot
find a source that has this exact result.
As such, we give a detailed proof.

\begin{lemma} \label{lm:shared Slater} \rm
In the above setting, let $C$ be a compact convex set in $\mathbb{R}^n$.  Suppose that
there exists a vector $\bar{x} \in \rint{C}$
such that $A\bar{x} = b$ and $g(\bar{x}) < 0$.
Then there exists a scalar $ \alpha > 0$ such that for every $x \in C \setminus D$, a vector $\wh{x} \in C$ exists
such that $r_q^{\, \prime}(x;\wh{x} - x) \, \leq \, - \alpha \, \| \, \wh{x} - x \, \|_2$.
\end{lemma}

\begin{proof}  
Let $\varepsilon > 0$ be such that the closed set
$\mathbb{B}_{\varepsilon}(\bar{x}) \, \triangleq  \, \left\{ \, x \, \in \, C \, \mid \, \| \, x - \bar{x} \, \|_2 \, \leq \, \varepsilon \, \right\}$
is contained in $\left\{ x \in C \, \mid \, g(x) < 0 \, \right\}$.  That such a scalar $\varepsilon$ exists is by the property of
the Slater point $\bar{x}$.  Let $\Sigma$ be the family of tuples $\sigma \in \{ 0, 1, -1 \}^p$ such that there exists $x \in C$
satisfying $\sgn(h(x)) = \sigma$, where $\sgn(h(x))$ is the $p$-vector whose $i$th component is equal to the sign of $h_i(x)$ if
$h_i(x) \neq 0$ and equal to 0 otherwise.  For each $\sigma \in \Sigma$, let
\[
C(\sigma) \triangleq \left\{ \, x \, \in \mathbb{B}_{\varepsilon}(\bar{x}) \, \left| \, \begin{array}{ll}
h_i(x) \, = \, 0 & \mbox{if $\sigma_i \, = \, 0$} \\ [5pt]
h_i(x) \, > \, 0 & \mbox{if $\sigma_i \, = \, 1$} \\ [5pt]
h_i(x) \, < \, 0 & \mbox{if $\sigma_i \, = \, 0$}
\end{array} \right. \right\}.
\]
This set $C(\sigma)$ must be nonempty because for any $\sigma \in \Sigma$, letting $x^{\, \sigma}$ be a vector in $C$ such that
$\sgn(h(x^{\, \sigma})) = \sigma$, by the linearlity of $h$ and the fact that $h(\bar{x}) = 0$, we have
\[
\sgn(h(\bar{x} + \tau ( x^{\, \sigma} - \bar{x} ))) = \, \sgn(h(x^{\, \sigma})) \, = \, \sigma,
\]
thus $\bar{x} + \tau ( x^{\, \sigma} - \bar{x} ) \in C(\sigma)$ for all $\tau \in (0,1]$ sufficiently small. Furthermore, since \( C \) is convex and \( \bar{x} \in \rint{C} \), for some \( \tau < 0 \) that is sufficiently close to \( 0 \), \( \bar{x} + \tau ( x^{\, \sigma} - \bar{x} ) \in C \) and \( \sgn(h(\bar{x} + \tau ( x^{\, \sigma} - \bar{x} ))) = - \sigma \). Therefore, \( \Sigma = -\Sigma \) and for every \( x\in C \), \( -\sgn h(x) \in \Sigma \). For each $\sigma \in \Sigma$, define
\[
\mu(\sigma) \, \triangleq \, \left\{ \begin{array}{ll}
1 & \mbox{if $\sigma = 0$} \\ [5pt]
\displaystyle{
\max_{x \in \cl C(\sigma)}
} \left[ \, \displaystyle{
\min_{i \, : \, \sigma_i \neq 0}
} \, \sigma_i \, h_i(x)\, \right] & \mbox{if $\sigma \neq 0$}.
\end{array} \right.
\]
where $\cl$ denotes the closure of a set.  Since $\displaystyle{
\min_{i \, : \, \sigma_i \neq 0}
} \, \sigma_i \, h_i(x)$ is a continuous function of $x$ and $\cl C(\sigma)$ is compact, it follows that maximum in $\mu(\sigma)$ is attained
if $\sigma \neq 0$.  Hence $\mu(\sigma)$ is a positive scalar for all $\sigma \in \Sigma$.  Let
\[
y(\sigma) \, \left\{ \begin{array}{lll}
\triangleq & \bar{x} & \mbox{if $\sigma = 0$} \\ [5pt]
\in & \displaystyle{
\operatornamewithlimits{\mbox{argmax}}_{x \in \cl C(\sigma)}
} \, \left[ \, \displaystyle{
\min_{i \, : \, \sigma_i \neq 0}
} \, \sigma_i \, h_i(x)\, \right] & \mbox{if $\sigma \neq 0$}.
\end{array} \right.
\]
Note that $g(y(\sigma)) < 0$ for all $\sigma \in \Sigma$.  Now define the scalar
\[
\eta \, \triangleq \, \displaystyle{
\min_{\sigma \in \Sigma}
} \, \left\{ \, \mu(\sigma), \, \displaystyle{
\min_{i \in [m]}
} \, \left[ \, -g_i(y(\sigma)) \, \right] \, \right\},
\]
which must be positive.
For $x \in C \setminus D$, $\sgn(h(x))$ is a nonzero sign tuple in the family $\Xi$.  
Let $\wh{x} \triangleq y(-\sgn(h(x)))$.  By the gradient inequality of convex functions 
and the definition of \( \eta \), we have the following:
\[ \begin{array}{clllll}
\nabla g_{i}(x)^{\mathrm{T}} ( \wh{x}-x ) & \leq & g_{i}( \wh{x} ) \, = \, g_{i}(y(- \sgn h(x))) & \leq & -\eta & \text{if } g_{i}(x) > 0 \\ [5pt]
(\,\sgn{h_{j}(x)}\,) \nabla h_{j}(x)^T ( \wh{x} - x) & \leq & (\,\sgn{h_{j}(x)}\,) \, h_{j}({y}(-\sgn h(x))) & \leq & -\eta &\text{if } h_{j}(x) \neq 0.
\end{array} \]
By the gradient formula (\ref{eq:lk-gradient}), we deduce,
\[ \begin{array}{lll}
\nabla r_q(x)^T( \, \wh{x} - x \, ) & = & \displaystyle{
\frac{1}{( \, r_q(x) \, )^{q-1}}
} \, \left[ \begin{array}{l}
\displaystyle{
\sum_{i \in [ m ]}
} \left[ \, \max( \, g_i(x), \, 0 \, ) \, \right]^{q-1} \, \nabla_{x^{\nu}} g_i(x) + \\ [0.25in]
\displaystyle{
\sum_{j \in [ p ]}
} \, | \, h_j(x) \, |^{q-1} \, ( \, \sgn h_{j}(x) \, ) \, \nabla_{x^{\nu}} h_{j}(x)
\end{array} \right]^T( \, \wh{x} - x \, ) \\ [0.45in]
& \leq & -\eta \, \left( \, \displaystyle{
\frac{r_{q-1}(x)}{r_q(x)}
} \, \right)^{q-1} \, \leq \, -\eta \, c_q^{q-1} \, \leq \, - \alpha \, \| \, \wh{x} - x \|_2,
\end{array} \]
where $c_q > 0$ is the constant such that $ c_q \, \| \bullet \|_q \leq \| \bullet \|_{q-1}$ and
$ \alpha \triangleq \displaystyle{
\frac{\eta \, c_q^{q-1}}{2 \, \displaystyle{
\max_{x \in C}
} \, \| \, x \, \|_2}
}$.
\end{proof}

The following exact penalization result for the shared-constrained GNEP is an immediate consequence of the above lemma
and part (b) of Theorem~\ref{th:exact penal shared}.  There is no need for a proof.

\begin{theorem} \label{th:exact penal shared finitely representable} \rm
Let each $C^{\, \nu}$ be a compact convex subset of $\mathbb{R}^{n_{\nu}}$ and
$D$ be given by (\ref{eq:shared finitely representable X}) with $h(x) = Ax - b$ being affine and
each $g_j(\bullet)$ being convex and differentiable for all $x^{-\nu} \in C^{\, -\nu}$.
Assume that each $\theta_{\nu}(\bullet,x^{-\nu})$ is convex and Lipschitz continuous on the set
$C^{\, \nu}$ with constant $\mbox{Lip}_{\theta} > 0$ for all $x^{-\nu} \in C^{\, -\nu}$.
Let $q \in ( 1, \infty )$ be given.  The following two statements hold.

\gap

(a) For every $\rho > 0$, the NEP $(C; \theta + \rho \, r_{q})$ has an equilibrium solution.

\gap

(b) Under the assumption of Lemma~\ref{lm:shared Slater}, $\bar{\rho} > 0$ exists such that for every $\rho > \bar{\rho}$, every equilibrium
solution of the NEP $(C;\theta + \rho \, r_q)$ is a equilibrium solution of the GNEP $(C,D;\theta)$.  \hfill $\Box$
\end{theorem}

We note that the property of boundedness and convexity of the sets $C^{\, \nu}$ is used to assert the existence of a solution to the penalized NEP for any $\rho > 0$.   Without such existence, the result is not meaningful, unless we have
a converse in place to argue that the GNEP will not have a solution in this case.  This is one aspect that we have de-emphasized in the discussion.

\gap

As a final remark of this section, we note that by following similar arguments,
we can establish an exact $\ell_1$-penalty result for the variational equilibria of the shared-constrained GNEP $(C,D;\theta)$
with a convex, finitely representable set $D$ under the standard Slater condition.  The proof would use the directional derivative
formula (\ref{eq:l1-dd}) for this penalty function instead of the gradient formula (\ref{eq:lk-gradient}).  The details are omitted.
Such a result provides a generalization of part (d) of the illustrative Example~\ref{eg:3-penalty-funs}.

\section{Non-Shared Finitely Representable Sets} \label{sec:finitely representable}

In spite of the considerable amount of literature on the exact penalization of a GNEP with coupled finitely representable constraints,
most existing results pertain to problems with only inequality constraints, see
e.g., \cite{facchinei-2010-penal-method, facchinei-2011-Lampariello, fukushima2011restricted, KanzowSteck16}.  The reason for
the absence of such an extended treatment to allow for equality constraints (even linear ones) is not exactly clear; one possible 
reason might be that the EMFCQ used in the aforementioned papers is not easily amenable to generalization to handle equality constraints. 
In this section, we extend the discussion in Subsection~\ref{subsec:shared finitely representable} to the
GNEP $(C,X;\theta)$ with non-shared coupled finitely representable constraints and show that an extended 
Slater-type CQ can be imposed to yield the exact penalization for such a GNEP.   
Specifically, we let the sets $X^{\nu}(x^{-\nu})$ be defined as follows: for each $x^{-\nu} \in C^{\, -\nu}$, 
\begin{equation} \label{eq:finitely representable X}
X^{\, \nu}(x^{-\nu}) \, \triangleq \,
\left\{ \, x^{\nu} \, \in \, \mathbb{R}^{n_{\nu}} \, \mid \,
g^{\nu}(x^{\nu},x^{-\nu}) \, \leq \, 0 \mbox{ and } h^{\nu}(x) \, = \, 0 \, \right\}
\end{equation}
with each $h^{\nu} : \mathbb{R}^n \to \mathbb{R}^{p_{\nu}}$ being an affine function and each
$g_j^{\nu} : \mathbb{R}^n \to \mathbb{R}$ for $j \in [m_{\nu}]$ being such that
$g_j^{\nu}(\bullet,x^{-\nu})$ is convex and differentiable for every $x^{-\nu} \in C^{\, -\nu}$.
We present a result that is a parallel of Theorem~\ref{th:exact penal shared finitely representable} for GNEP \( (C, X; \theta) \) and then we connect the assumptions used in the result with the commonly used EMFCQ. 
Let
\begin{equation} \label{eq:residual finitely representable}
r_{\nu;q}(x) \, \triangleq \, \left\| \, \left( \, \begin{array}{c} 
\max( \, g^{\nu}(x), \, 0 \, ) \\ [3pt]
h^{\nu}(x)
\end{array} \right) \, \right\|_q, \epc x \, \in \, \mathbb{R}^n
\end{equation}
be the $\ell_q$ residual function for the set $X^{\nu}(x^{-\nu})$. 
This function $r_{\nu;q}(\bullet,x^{-\nu})$ is continuously differentiable on the complement of the set $X^{\nu}(x^{-\nu})$ for
all $x^{-\nu}$.  Let $r_{X;q} \triangleq \left( r_{\nu;q} \right)_{\nu \in [ N ]}$.
The penalized Nash game, NEP $(C; \theta + \rho \, r_{X;q})$, consists of the following family of optimization problems:
\[
\left\{ \, \displaystyle{
\operatornamewithlimits{\mbox{minimize}}_{x^{\nu} \, \in \, C^{\, \nu}}
} \ \theta_{\nu}(x^{\nu},x^{-\nu}) + \rho \, r_{\nu;q}(x^{\nu},x^{-\nu}) \, \right\}_{\nu=1}^N,
\]
for a given penalty parameter $\rho > 0$.

\gap

\begin{theorem} \label{th:finitely representable} \rm
Let each $C^{\, \nu}$ be a compact convex subset of $\mathbb{R}^{n_{\nu}}$.
Let each \( X^{\nu} \) be defined as in \eqref{eq:finitely representable X} with \( h^{\nu} \) being affine and each $g_j^{\nu}(\bullet,x^{-\nu})$ be convex and differentiable for all $x^{-\nu} \in C^{\, -\nu}$.
Assume that each $\theta_{\nu}(\bullet,x^{-\nu})$ is convex and Lipschitz continuous on the set
$C^{\, \nu}$ with constant $\mbox{Lip}_{\theta} > 0$ for all $x^{-\nu} \in C^{\, -\nu}$.
Let $q \in ( 1, \infty )$ be given.  The following three statements hold.

\gap

(a) For every $\rho > 0$, the NEP $(C; \theta + \rho \, r_{X;q})$ has an equilibrium solution.

\gap

%

(b) If there exists $\alpha > 0$ such that
for every $x \in C$ not in the graph of the multifunction $X$, there exists $\wh{x} \in C$ satisfying
for some $\wh{\nu} \in [ N ]$ with $x^{\wh{\nu}} \not\in X^{\wh{\nu}}(x^{-\wh{\nu}})$,

\gap

$\bullet $  for all $i \in [m_{\nu}]$,
$\nabla_{x^{\wh{\nu}}} g^{\wh{\nu}}_i(x^{\, \wh{\nu}},x^{-\wh{\nu}})^T( \, \wh{x}^{\, \wh{\nu}} - x^{\wh{\nu}} \, ) \, \leq -\alpha$
if $g_i^{\wh{\nu}}(x) > 0$, and

\gap

$\bullet $  for all $j \in [p_{\nu} ]$,
$\left[ \, \sgn h_j^{\wh{\nu}}(x) \, \right] \nabla_{x^{\wh{\nu}}} h^{\wh{\nu}}_j(x^{\, \wh{\nu}},x^{-\wh{\nu}})^T
( \, \wh{x}^{\, \wh{\nu}} - x^{\wh{\nu}} \, ) \, \leq -\alpha$
if $h_{j}^{\wh{\nu}}(x) \neq 0$,

\gap

then $\bar{\rho} > 0$ exists such that for every $\rho > \bar{\rho}$, every equilibrium
solution of the NEP $(C; \theta + \rho \, r_{X;q})$ is a equilibrium solution of the GNEP $(C,X;\theta)$.
\end{theorem}

\begin{proof}  Assertion (a) follows from a standard existence result for a NEP. 
In order to prove statement (b), let $x^*$ be an equilibrium solution of the NEP $(C; \theta + \rho \, r_{X;q})$
for some $\rho > 0$ which will be specified momentarily.
By part (a) of Proposition~\ref{pr:NEP feasible yields GNEP}, it suffices to show that $x^*$ belongs to the
graph of the multifunction $X$.  
Assume the contrary; let $\wh{x}$ and $\wh{\nu}$ be the vector and player label, respectively, associated with $x^*$ 
as stipulated by the assumption in this part.
By the gradient formula (\ref{eq:lk-gradient}), we deduce,
\[ \begin{array}{l}
\nabla_{x^{\wh{\nu}}} r_{\wh{\nu};q}(x^*)^T( \, \wh{x}^{\, \wh{\nu}} - x^{*,\wh{\nu}} \, ) \\ [5pt]
= \, \displaystyle{
\frac{1}{r_{\wh{\nu};q}^{q-1}(x^*)}
} \, \left[ \begin{array}{l}
\displaystyle{
\sum_{i \in [ m_{\wh{\nu}} ]}
} \left[ \, \max( \, g_i^{\wh{\nu}}(x^*), \, 0 \, ) \, \right]^{q-1} \, \nabla_{x^{\wh{\nu}}} g_{i}^{\wh{\nu}}(x^*) + \\ [0.25in]
\displaystyle{
\sum_{j \in [ p_{\wh{\nu}} ]}
} \, | \, h_{j}^{\wh{\nu}}(x^*) \, |^{q-1} \, ( \, \sgn h_{j}^{\wh{\nu}}(x^*) \, ) \, \nabla_{x^{\wh{\nu}}} h_{j}^{\wh{\nu}}(x^*)
\end{array} \right]^T( \, \wh{x}^{\, \wh{\nu}} - x^{*,\wh{\nu}} \, ) \\ [0.45in]
\leq \, -\alpha \, \left( \, \displaystyle{
\frac{r_{\wh{\nu},q-1}(x^*)}{r_{\wh{\nu},q}(x^*)}
} \, \right)^{q-1} \, \leq \, -\alpha \, c_q^{q-1},
\end{array} \]
where $c_q > 0$ is the constant such that $ c_{q} \,\| \bullet \|_q \leq \| \bullet \|_{q-1}$.
Since
\[
x^{*,\wh{\nu}} \, \in \, \displaystyle{
\operatornamewithlimits{\mbox{minimize}}_{x^{\wh{\nu}} \, \in \, C^{\, \wh{\nu}} \, \cap \, X^{\wh{\nu}}(x^{*,-\wh{\nu}})}
} \ \theta_{\wh{\nu}}(x^{\wh{\nu}},x^{*,-\wh{\nu}}) + \rho \, r_{\wh{\nu};q}(x^{\wh{\nu}},x^{*,-\wh{\nu}}),
\]
we deduce
\[ \begin{array}{lll}
0 & \leq &
\theta_{\wh{\nu}}(\bullet,x^{*,-\wh{\nu}})^{\, \prime}(x^{*,\wh{\nu}};\wh{x}^{\, \wh{\nu}} - x^{*,\wh{\nu}}) +
\rho \, \nabla_{x^{\wh{\nu}}} r_{\wh{\nu};q}(x^*)^T( \, \wh{x}^{\, \wh{\nu}} - x^{*,\wh{\nu}} \, )
\\ [0.1in]
& \leq & \mbox{Lip}_{\theta} \, \| \, \wh{x}^{\, \wh{\nu}} - x^{*,\wh{\nu}} \, \|_2 - \rho \, \alpha \, c_q^{q-1}
\, \leq \, \left[ \, \mbox{Lip}_{\theta} - \rho \, \wh{\alpha} \, \right] \, \| \, \wh{x}^{\, \wh{\nu}} - x^{*,\wh{\nu}} \, \|_2.
\end{array} \]
where $\wh{\alpha} \triangleq \displaystyle{
\frac{\alpha \, c_q^{q-1}}{2 \, \displaystyle{
\max_{x \in C}
} \, \| \, x \, \|_2}
}$.
Provided that $\rho > \wh{\alpha}$, the above inequality yields a contradiction.
\end{proof}

We close this paper by showing that the commonly used EMFCQ is a more restrictive assumption than the one imposed in
part (b) of Theorem~\ref{th:finitely representable}.  Besides the fact that the EMFCQ does not admit equality constraints, there
is another subtle restriction in it; namely, included in this CQ is a condition on a feasible vector in the graph of the multifunction $X$,
whereas the condition in part (b) of Theorem~\ref{th:finitely representable} does not contain such a restriction.  
We state the EMFCQ for a convex set $C$ as follows, see also \cite[Definition~4.3]{CominettiFacchineiLasserre12}. 

\gap

{\bf (EMFCQ)}  This CQ is said to hold for the GNEP \( (C,X; \theta) \), where each
\begin{equation}
\label{eq:X-g}
X^{\nu}(x^{-\nu}) \, = \, \left\{ \, x^{\nu} \, \mid \, g^{\nu}(x^{\nu},x^{-\nu}) \, \leq \, 0 \, \right\}
\end{equation}
with each $g_i^{\nu}(\bullet,x^{-\nu})$ being convex and continuously differentiable, if for every $x \in C$ and \( \nu \in [N] \), there exists a vector $\wh{x} \in C$ such that
\[
\nabla_{x^{\nu}} g^{\nu}_i(x^{\nu},x^{-\nu})^T( \wh{x}^{\nu} - x^{\nu} ) \, < \, 0 \epc \mbox{if $g_i^{\nu}(x^{\nu}, x^{-\nu}) \geq 0$}.
\]
Note that this CQ restricts even a vector $x \in C$ belonging to the graph of the multifunction $X$, i.e., even
for a feasible vector $x$ to the GNEP.  The reason
to include such vectors in the CQ is so that the Karush-Kuhn-Tucker conditions become valid if $x$ turns out be an equilibrium solution of the GNEP.  It turns out
that as far as exact penalization is concerned, the validity of the CQ can be somewhat relaxed.
We have the following result.

\begin{proposition} \label{prop: emfcq2strong-descent} \rm
Let each \( C^{\, \nu} \) be a convex compact subset of $\mathbb{R}^{n_{\nu}}$ and each \( X^{\nu} \) be defined as in \eqref{eq:X-g} with each $g_{i}^{\nu}(\bullet,x^{-\nu})$ being convex and continuously differentiable.  If EMFCQ holds for the GNEP \( (C, X; \theta) \) , then the assumption in part (b) of Theorem is valid.
\end{proposition}

\begin{proof} The proof is by contradiction.  We first restate the assumption in part (b) in Theorem~\ref{th:finitely representable} in a
set-theoretic form that is more conducive to obtain the contradiction.  For a given $x \in C$,
let
\[
{\cal I}_{\nu}(x) \, \triangleq \, \left\{ \, i \, \in \, \{ \, 1, \cdots, m_{\nu} \, \} \, \mid \, g_i^{\, \nu}(x^{\nu},x^{-\nu}) \, > \, 0 \right\}
\]
be the index set of violated constraints in $X^{\nu}(x^{-\nu})$ by $x^{\nu}$.   The statement that $x^{\nu} \not\in X^{\nu}(x^{-\nu})$
is equivalent to ${\cal I}_{\nu}(x) \neq \emptyset$.  Also let
\[
\wh{\cal I}(x) \, \triangleq \, \left\{ \, \nu \, \in \, [ N ] \, \mid \, {\cal I}_{\nu}(x) \, \neq \, \emptyset \, \right\}
\]
be the player labels $\nu$ for which there is at least one constraint in $X^{\nu}(x^{-\nu})$ that is not satisfied by $x^{\nu}$.
Thus $x \not\in \mbox{ graph}(X)$ if and only if $\wh{\cal I}(x) \neq \emptyset$.
Coupled with a $y \in C$, define
the quantity
\[
\alpha_{\nu}(x,y) \, \triangleq \, \displaystyle{
\max_{i \in {\cal I}_{\nu}(x)}
} \ \nabla_{x^{\nu}} g_i^{\, \nu}(x^{\nu},x^{-\nu})^T( y^{\, \nu} - x^{\nu} ).
\]
By convention, the maximum over the empty set is equal to $-\infty$.  For a given vector $x \in C$ not in $\mbox{graph}(X)$ and
a player label $\nu \in \wh{\cal I}(x)$, the quantity $\alpha_{\nu}(x,\bullet)$ is a continuous function on $C$.  Thus it is justified
to define the scalar
\[
\gamma \, \triangleq \, \displaystyle{
\sup_{x \in C \setminus {\rm graph}(X)}
} \, \left[ \, \displaystyle{
\min_{\nu \in \wh{\cal I}(x)}
} \, \left( \, \displaystyle{
\min_{y \in C}
} \, \alpha_{\nu}(x,y) \, \right) \, \right],
\]
where we use sup instead of max because the latter may not be attained.
Now, the assumption in part (b) in Theorem~\ref{th:finitely representable} becomes $\gamma < 0$.  So to prove the proposition by
contradiction, suppose $\gamma \geq 0$.    Then there exist a sequence of vectors $\{ x^k \} \subset C$ and a sequence of positive
scalars $\{ \varepsilon_k \} \downarrow 0$ such that for every $k$,
$x^k \not\in \mbox{ graph}(X)$ and there exists $\nu_k \in \wh{\cal I}(x^k)$ such that $\displaystyle{
\min_{y \in C}
} \, \alpha_{\nu_k}(x^k,y) \geq -\varepsilon_k$.  Without loss of generality, we may assume that the sequence $\{ x^k \}$ converges
to a vector $x^{\infty} \in C$.
Since there are only finitely many players, it follows that there exists a subsequence $\{ x^k \}_{k \in \kappa}$ and a player label
$\bar{\nu}$ such that for all $k \in \kappa$, we have $\bar{\nu} \in \wh{\cal I}(x^k)$ and $\displaystyle{
\min_{y \in C}
} \, \alpha_{\bar{\nu}}(x^k,y) \geq -\varepsilon_k$.  Hence for all $y \in C$, $\alpha_{\bar{\nu}}(x^k,y) \geq -\varepsilon_k$
for all $k \in \kappa$.  Passing to the limit $k ( \in \kappa ) \to \infty$, we deduce that
$\displaystyle{
\min_{y \in C}
} \, \alpha_{\bar{\nu}}(x^{\infty},y) \geq 0$.  Since $\bar{\nu} \in \wh{\cal I}(x^k)$ for all $k \in \kappa$, it follows that
$g_i^{\, \bar{\nu}}(x^{\infty}) \geq 0$ for all $i \in {\cal I}_{\bar{\nu}}(x^{\infty})$.  Thus we have obtained the following
\[
\forall \, y \, \in \, C \ \exists \, i \mbox{ such that } g_i^{\, \bar{\nu}}(x^{\infty}) \, \geq \, 0
\mbox{ and } \nabla_{x^{\nu}} g_i^{\, \bar{\nu}}(x^{\infty,\bar{\nu}},x^{\infty,-\bar{\nu}})^T( y^{\, \bar{\nu}} - x^{\infty,\bar{\nu}} ) \, \geq \, 0.
\]
But this contradicts the EMFCQ applied to $x^{\infty}$ which asserts
that there exists $\wh{x} \in C$ such that $\nabla_{x^{\nu}} g_i^{\, \nu}(x^{\infty,\nu},x^{\infty,-\nu})^T( \wh{x}^{\, \nu} - x^{\infty,\nu} ) < 0$
whenever $g_i^{\, \nu}(x^{\infty,\nu},x^{\infty,-\nu}) \geq 0$.
\end{proof}
%
%
%
%

{\bf Conclusion.}  In this paper, we have provided a comprehensive exact penalization theory for the GNEP, by putting together a systematic study
of the subject that begins with an illustrative example and ends with a pair of results for the cases of finitely representable constraint sets.  In the process,
we have clarified the role of error bounds and constraint qualifications in the developed theory; this is something that has not been sufficiently
emphasized in the existing literature of the GNEP.  How these exact penalization results can be put to use in the development of improved algorithms for
computing an equilibrium solution of the GNEP remains to be investigated.  

\gap

{\bf Acknowledgement.}  The authors are grateful to Christian Kanzow and Francisco Facchinei for some helpful comments on a draft of this paper.


\begin{thebibliography}{999}

\bibitem{arrow1954existence}
Kenneth~J.\ Arrow and Gerard Debreu.
\newblock Existence of an equilibrium for a competitive economy.
\newblock {\sl Econometrica: Journal of the Econometric Society} 22(3): 265--290, 1954.

\bibitem{Auslender76}
Alfred Auslender.
\newblock {\sl Optimisation M\'ethodes Num\'eriques}.
Masson, Paris, France, 1976.

\bibitem{aussel2017sufficient}
Didier Aussel and Simone Sagratella
\newblock Sufficient conditions to compute any solution of a quasivariational inequality via a variational inequality.
\newblock {\sl Mathematical Methods of Operations Research} 85(1), 3--18, 2017.

\bibitem{BanDessoukyPang2018}
Jeff X.\ Ban, Jong-Shi Pang, and Maged Dessouky.
\newblock A general equilibrium model for transportation systems with e-hailing services and flow congestion.
\newblock Manuscript, Department of Industrial and Systems Engineering, University of Southern California.
September 2018.

\bibitem{bassanini2002allocation}
A~Bassanini, A~La~Bella, and A~Nastasi.
\newblock Allocation of railroad capacity under competition: a game theoretic
  approach to track time pricing.
\newblock In {\sl  Transportation and Network Analysis: Current Trends}, pages
  1--17. Springer, 2002.

\bibitem{bauschke-1999-stron-conic}
Heinz~H. Bauschke, Jonathan~M. Borwein, and Wu~Li.
\newblock Strong conical hull intersection property, bounded linear regularity,
  Jameson's property (g), and error bounds in convex optimization.
\newblock {\sl  Mathematical Programming} 86(1):135--160, 1999.

\bibitem{breton2006game}
Michele Breton, Georges Zaccour, and Mehdi Zahaf.
\newblock A game-theoretic formulation of joint implementation of environmental projects.
\newblock {\sl European Journal of Operational Research} 168(1):221--239, 2006.

\bibitem{ByrdNocedalWaltz08}
Richard H.\ Byrd, Jorge Nocedal, and Richard A.\ Waltz.
\newblock Steering exact penalty methods for nonlinear programming.
\newblock {\sl Optimization Methods and Software} 23(2): 197--213, 2008.

\bibitem{ByrdLopezCalvaNocedal12}
Richard H.\ Byrd, G Lopez-Calva, and Jorge Nocedal.
\newblock A line search exact penalty method using steering rules.
\newblock {\sl Mathematical Programming} 133(1-2): 39--73, 2012.

\bibitem{ByrdNocedalWaltz06}
Richard H.\ Byrd, Jorge Nocedal, and Richard A.\ Waltz.
\newblock Knitro: An integrated package for nonlinear optimization.
\newblock In G.\ Di Pillo and M.\ Roma, editors.
{\sl Large-Scale Nonlinear Optimization}, pages 35--59, Springer, Boston, 2006.

\bibitem{Clarke83}
Francis H.\ Clarke.
\newblock {\sl Optimization and Nonsmooth Analysis}.
\newblock Classics in Applied Mathematics, Volume 5, {\sl Society for Industrial and Applied Mathematics}, 1990.
[Reprint from John Wiley Publishers, New York 1983.]

\bibitem{CominettiFacchineiLasserre12}
Roberto Cominetti, Francisco Facchinei and Jean-Bernard Lasserre.  Authors.
\newblock {\sl Modern Optimization Modeling Techniques}.
\newblock Birkhauser Springer Basel, 2012.
10.1007/978-3-0348-0291-8.


\bibitem{debreu-1952-social-equil}
Gerard Debreu.
\newblock A social equilibrium existence theorem.
\newblock {\sl Proceedings of the National Academy of Sciences} 38(10):886--893, 1952.

\bibitem{demyanov1998exact}
Vladimir~F.\ Demyanov, Gianni Di~Pillo, and Francisco Facchinei.
\newblock Exact penalization via Dini and Hadamard conditional derivatives.
\newblock {\sl  Optimization Methods and Software} 9(1-3):19--36, 1998.

\bibitem{di1989exact}
Gianni Di~Pillo and F~Facchinei.
\newblock Exact penalty functions for nondifferentiable programming problems.
\newblock In {\sl Nonsmooth optimization and related topics}, pages 89--107.
  Springer, 1989.

\bibitem{di1992regularity}
Gianni Di~Pillo and Francisco Facchinei.
\newblock Regularity conditions and exact penalty functions in Lipschitz
  programming problems.
\newblock {\sl Nonsmooth optimization methods and applications}, pages
  107--120, 1992.

\bibitem{di1995exact}
Gianni Di~Pillo and Francisco Facchinei.
\newblock Exact barrier function methods for Lipschitz programs.
\newblock {\sl Applied Mathematics and Optimization} 32(1):1--31, 1995.

\bibitem{di1989b-exact}
G. Di Pillo and L. Grippo
\newblock Exact penalty functions in constrained optimization.
\newblock {\sl SIAM Journal on control and optimization} 27(6), 1333--1360, 1989.

\bibitem{dreves2017computing}
Axel Dreves.
\newblock Computing all solutions of linear generalized Nash equilibrium  problems.
\newblock {\sl Mathematical Methods of Operations Research} 85(2):207--221, 2017.

\bibitem{DrevesFacchineiKanzowSagratella11}
Axel Dreves, Francisco Facchinei, Christian Kanzow, and Simone Sagratella.
\newblock On the solution of the KKT conditions of generalized Nash equilibrium problems.
\newblock {\sl SIAM Journal on Optimization} 21(3):1082--1108, 2011.

\bibitem{FacchineiFischerPicciali07}
Francisco Facchinei, Andreas Fischer, and Veronica Piccialli.
\newblock On generalized Nash games and variational inequalities.
\newblock {\sl  Operations Research Letters} 35:159--164, 2007.

\bibitem{facchinei09GNEP}
Francisco Facchinei and Christian Kanzow.
\newblock Generalized Nash equilibrium problems.
\newblock {\sl  Annals of Operations Research} 175(1):177--211, 2007.

\bibitem{facchinei-2010-penal-method}
Francisco Facchinei and Christian Kanzow.
\newblock Penalty methods for the solution of generalized Nash equilibrium
  problems.
\newblock {\sl  SIAM Journal on Optimization} 20(5):2228--2253, 2010.

\bibitem{facchinei-2011-Lampariello}
Facchinei Facchinei and Lorenzo Lampariello.
\newblock Partial penalization for the solution of
generalized Nash equilibrium problems.
\newblock {\sl Journal of Global Optimization} 50(1): 39--50, 2011.

\bibitem{facchinei2007finite}
Francisco Facchinei and Jong-Shi Pang.
\newblock {\sl  Finite-dimensional variational inequalities and complementarity
  problems}.
\newblock Springer Science \& Business Media, 2003.

\bibitem{facchinei2006exact}
Francisco Facchinei and Jong-Shi Pang.
\newblock Exact penalty functions for generalized Nash problems.
\newblock In {\sl  Large-scale nonlinear optimization}, pages 115--126.
  Springer, 2006.

\bibitem{facchinei201012}
Francisco Facchinei and Jong-Shi Pang.
\newblock Nash equilibria: the variational approach.
\newblock Chapter 12 in Daniel P.\ Palomar and Yonica C.\ Eldar, editors.
{\sl Convex optimization in signal processing and communications},
Cambridge University Press, pages 443--493, 2010.

\bibitem{facchinei2011computation}
Francisco Facchinei and Simone Sagratella.
\newblock On the computation of all solutions of jointly convex generalized
  Nash equilibrium problems.
\newblock {\sl  Optimization Letters}, 5(3):531--547, 2011.

\bibitem{FiaccoMcCormick68}
Anthony V.\ Fiacco and Garth P.\ McCormick.
\newblock {\sl Nonlinear Programming: Sequential Unconstrained Minimization Techniques}.
\newblock Classics in Applied Mathematics, Volume 4, {\sl Society for Industrial and Applied Mathematics}, 1990.
[Reprint from John Wiley Publishers, New York 1968.]

\bibitem{fischer2014generalized}
Andreas Fischer, Markus Herrich, and Klaus Sch{\"o}nefeld.
\newblock Generalized Nash equilibrium problems-recent advances and challenges.
\newblock {\sl  Pesquisa Operacional}, 34(3):521--558, 2014.

\bibitem{fukushima2011restricted}
Masao Fukushima.
\newblock Restricted generalized Nash equilibria and controlled penalty
  algorithm.
\newblock {\sl  Computational Management Science}, 8(3):201--218, 2011.

\bibitem{gwinner1983penalty}
J~Gwinner.
\newblock On the penalty method for constrained variational inequalities.
\newblock {\sl Optimization: Theory and Algorithms} 86(1): 197--211, 1983.

\bibitem{HanMangasarian79}
S.P.\ Han and O.L.\ Mangasarian.
\newblock Exact penalty functions in nonlinear programming.
\newblock {\sl Mathematical Programming} 17(1): 251--269, 1979.

\bibitem{harker1991generalized}
Patrick~T Harker.
\newblock Generalized Nash games and quasi-variational inequalities.
\newblock {\sl European Journal of Operational Research} 54(1): 81--94, 1991.

\bibitem{hobbs2007Nash}
Benjamin~F.\ Hobbs and Jong-Shi Pang.
\newblock Nash-cournot equilibria in electric power markets with piecewise
  linear demand functions and joint constraints.
\newblock {\sl Operations Research} 55(1): 113--127, 2007.

\bibitem{KanzowSteck16}
Christian Kanzow and Daniel Steck.
\newblock Augmented Lagrangian and exact penalty methods for generalized Nash equilibrium
problems.
\newblock {\sl SIAM Journal on Optimization} 26(4): 2034–-2058.

\bibitem{KanzowSteck18}
Christian Kanzow and Daniel Steck.
\newblock Augmented Lagrangian and exact penalty methods for quasi-variational inequalities.
\newblock {\sl Computational Optimization and Applications} 69(3): 801–-824, 2018.

\bibitem{Klatte1999}
Diethard Klatte and Wu~Li.
\newblock Asymptotic constraint qualifications and global error bounds for
  convex inequalities.
\newblock {\sl Mathematical Programming} 84(1): 137--160, 1999.

\bibitem{Krawczyk07}
Jacek B.\ Krawczyk.
\newblock Numerical solutions to coupled-constraint (or generalised) Nash equilibrium problems.
\newblock {\sl Computational Management Science} 4(2): 183--204 2007.

\bibitem{Krawczyk05}
Jacek B.\ Krawczyk.
\newblock Coupled constraint Nash equilibria in environmental games.
\newblock {\sl Resource and Energy Economics} 27(2): 157--181, 2005.

\bibitem{KrawczykUryasev00}
Jacek B.\ Krawczyk and Stanislav Uryasev.
\newblock Relaxation algorithms to find Nash equilibria with economic applications.
\newblock {\sl Environmental Modeling and Assessment} 5: 63–-73 (2000).

\bibitem{KulkarniShanbhag12}
Ankur A.\ Kulkarni and Uday V.\ Shanbhag.
\newblock On the variational equilibrium as a refinement of the generalized Nash equilibrium.
\newblock {\sl Automatica} 48(1) 45--55, 2012.

\bibitem{lewis1998error}
Adrian~S.\ Lewis and Jong-Shi Pang.
\newblock Error bounds for convex inequality systems.
\newblock In {\sl Generalized convexity, generalized monotonicity: recent
  results}, pages 75--110. Springer, 1998.

\bibitem{nabetani2011parametrized}
Koichi Nabetani, Paul Tseng, and Masao Fukushima.
\newblock Parametrized variational inequality approaches to generalized Nash
  equilibrium problems with shared constraints.
\newblock {\sl  Computational Optimization and Applications} 48(3):423--452,
  2011.

\bibitem{ng2004-regul-their}
Kung~Fu Ng and Wei~Hong Yang.
\newblock Regularities and their relations to error bounds.
\newblock {\sl  Mathematical Programming} 99(3):521--538, 2004.

\bibitem{NikaidoIsoda55}
Hukukane Nikaido and Kazuo Isoda.
\newblock Note on noncooperative convex games.
\newblock {\sl Pacific Journal of Mathematics} 5(Suppl.~1): 807–-815, 1955.

\bibitem{pang1997error}
Jong-Shi Pang.
\newblock Error bounds in mathematical programming.
\newblock {\sl  Mathematical Programming} 79(1-3):299--332, 1997.

\bibitem{pang-2009-quasi-variat}
Jong-Shi Pang and Masao Fukushima.
\newblock Quasi-variational inequalities, generalized Nash equilibria, and
  multi-leader-follower games.
\newblock {\sl  Computational Management Science} 2: 21-–56, 2005.
[Erratum: same journal 6: 373-–375, 2009.]

\bibitem{pang2008distributed}
Jong-Shi Pang, Gesualdo Scutari, Francisco Facchinei, and Chaoxiong Wang.
\newblock Distributed power allocation with rate constraints in gaussian
  parallel interference channels.
\newblock {\sl  IEEE Transactions on Information Theory} 54(8):3471--3489,
  2008.

\bibitem{pang2010design}
Jong-Shi Pang, Gesualdo Scutari, Daniel~P.\ Palomar, and Francisco Facchinei.
\newblock Design of cognitive radio systems under temperature-interference
  constraints: A variational inequality approach.
\newblock {\sl  IEEE Transactions on Signal Processing} 58(6):3251--3271, 2010.

\bibitem{Robinson91}
Stephen M.\ Robinson.
\newblock An implicit-function theorem for a class of nonsmooth functions.
\newblock {\sl Mathematics of Operations Research} 16(2): 292--309, 1991.

\bibitem{rosen1965existence}
J.\ Ben Rosen.
\newblock Existence and uniqueness of equilibrium points for concave n-person
  games.
\newblock {\sl Econometrica: Journal of the Econometric Society} 33(3):520--534, 1965.

\bibitem{SchiroPangShanbhag13}
Dand A.\ Schiro, Jong-Shi Pang, and Uday V.\ Shanbhag.
\newblock On the solution of affine generalized Nash equilibrium problems with shared constraints by Lemke's method.
\newblock {\sl Mathematical Programming, Series A} 146(1):1--46, 2013.

\bibitem{stein2018noncooperative}
Oliver Stein and Nathan Sudermann-Merx.
\newblock The noncooperative transportation problem and linear generalized Nash
  games.
\newblock {\sl European Journal of Operational Research} 266(2):543--553,
  2018.

\bibitem{Tseng02}
Paul Tseng.
\newblock Convergence of a block coordinate descent method for nondifferentiable minimization.
\newblock {\sl Journal of Optimization Theory and Applications} 109(3): 475--494, 2002.

\bibitem{von2009optimization}
Anna von Heusinger and Christian Kanzow.
\newblock Optimization reformulations of the generalized Nash equilibrium
  problem using nikaido-isoda-type functions.
\newblock {\sl Computational Optimization and Applications} 43(3):353--377,
  2009.

\bibitem{jing1999spatial}
Jing-Yuan Wei and Yves Smeers.
\newblock Spatial oligopolistic electricity models with Cournot generators and
  regulated transmission prices.
\newblock {\sl Operations Research} 47(1):102--112, 1999.




\end{thebibliography}
\end{document}